\documentclass[journal,onecolumn]{IEEEtran}

\ifCLASSINFOpdf
\else
\fi
\usepackage[noadjust]{cite}
\usepackage[T1]{fontenc}
\usepackage{enumitem}
\usepackage{amsmath,amssymb,amsthm}
\newtheorem{thm}{Theorem}[section]
\newtheorem{lemma}[thm]{Lemma}
\newtheorem{prop}[thm]{Proposition}
\newtheorem{constr}{Construction}

\newtheorem{cor}[thm]{Corollary}

\newtheorem{exampl}[thm]{Example}
\newtheorem{prob}[thm]{Problem}
\newtheorem{dfn}[thm]{Definition}

\DeclareMathOperator{\ind}{ind}

\DeclareMathOperator{\ord}{ord}
\DeclareMathOperator{\Cay}{Cay}
\DeclareMathOperator{\Gal}{Gal}
\DeclareMathOperator{\wt}{wt}
\usepackage{tkz-graph}

\begin{document}
\title{New Results on Limited Magnitude Error Correcting Codes}
\author{
	Zhiyu Yuan\IEEEauthorrefmark{1}, \IEEEauthorblockN{Tingting Chen\IEEEauthorrefmark{2}, Rongquan Feng\IEEEauthorrefmark{3} and Gennian Ge\IEEEauthorrefmark{4}}

	\IEEEauthorblockA{\IEEEauthorrefmark{1}School of Mathematical Sciences, Peking University, Beijing 100871, China. Email: yzhiyu\_pku@pku.edu.cn}
	
	\IEEEauthorblockA{\IEEEauthorrefmark{2}Institute of Mathematics and Interdisciplinary Sciences, Xidian University, Xi'an, 710071, China. Email: ttchenxu@mail.ustc.edu.cn}
	
	\IEEEauthorblockA{\IEEEauthorrefmark{3}School of Mathematical Sciences, Peking University, Beijing 100871, China. Email: fengrq@math.pku.edu.cn}

	\IEEEauthorblockA{\IEEEauthorrefmark{4}School of Mathematical Sciences, Capital Normal University, Beijing 100048, China. Email: gnge@zju.edu.cn}

}

\maketitle

\begin{abstract}
    This paper investigates the existence, construction and classification of limited magnitude error-correcting codes, with a focus on splitter sets and their connections to group splittings. We establish new nonexistence results for quasi-perfect splitter sets and provide a complete classification of quasi-perfect $B[0,3](n)$ splitter sets in both singular and nonsingular cases. Furthermore, we derive improved lower bounds for the size of maximal $B[0,3](q)$ sets by investigating Cayley graphs, where $q$ is a prime. We also provide existence criteria for perfect $B[0,6](q)$ splitter sets and quasi-perfect $B[-4,4](2p)$ sets for prime $p$. For perfect burst-correcting codes, we develop a general construction framework, and prove the existence of infinite families of $(k_2,k_1)$-limited-magnitude cyclic $b$-burst-correcting codes for $k_1+k_2\le 4$ and arbitrary burst length $b$. We further provide sufficient existence conditions for general parameters $k_1$ and $k_2$. Our results combine algebraic, combinatorial, and number-theoretic methods to advance the understanding of codes tailored for flash memory and related storage systems.
\end{abstract}


\IEEEpeerreviewmaketitle

\section{Introduction}
Flash memories are non-volatile, high density and low cost memories that have applications in many areas of modern life. For a higher density of flash memories, \emph{multilevel memory cells} were introduced, which can store $q$ levels, and common flash error mechanisms induce errors whose magnitudes (i.e., the number of level changes) are small. This setting stimulated research into specialized error-correcting codes for flash memory storage, called limited magnitude error correcting codes, first proposed in \cite{cassuto2010codes} and \cite{klove2011systematic}. In this model, each symbol is an element in $\mathbb{Z}$.\footnote{While the original model adopted $\mathbb{Z}_q$ as the underlying alphabet, this paper, following \cite{wei2022tilings}, uses $\mathbb{Z}$ as the alphabet. } Errors are modeled as bounded additive perturbations: for a codeword $\mathbf{c} = (c_1, \ldots, c_n)$, a symbol $c_i$ may be distorted to $c_i + \lambda$, where $\lambda \in [-k_1, k_2] = \{-k_1, -k_1+1, \ldots, k_2\}$ for some nonnegative integers $k_1$ and $k_2$ with $k_2 > 0$.
The error is \emph{asymmetric} if $k_1=0$; when $k_1=k_2$, the error is \emph{symmetric}. In general, such errors are called \emph{unbalanced} as in \cite{yari2013some}.

Besides flash memories, such codes also find wide applications in high-density magnetic recording channels 
\cite{kuznetsov1993coding, levenshtein1993perfect}
and DNA-based storage systems 
\cite{jain2020coding,wei2021lattice}.

To give a formal definition for these codes, we first define the error ball
\[\mathcal{B}(n,t,k_2,k_1)=\left\{\mathbf{e}=(e_1,e_2,\ldots,e_n):\begin{array}{ll}e_i\in [-k_1,k_2],\\\wt(e)\le t\end{array}\right\}%
	,\]
that is, for every element of $\mathcal{B}(n,t,k_2,k_1)$, all its coordinates are zero, except for possibly at most $t$ coordinates being in $[-k_1,k_2]\setminus\{0\}=:[-k_1,k_2]^*$. Our goal is to construct a code $C\subset\mathbb{Z}^n$ such that for all $\mathbf{x}\in C$, $\mathbf{x}+\mathcal{B}(n,t,k_2,k_1)$ are disjoint. It is equivalent to saying that for all $\mathbf{e}\in \mathcal{B}(n,t,k_2,k_1)$, $\mathbf{e}+C$ are disjoint. We refer to such codes as \emph{$(k_2,k_1)$-limited magnitude $t$-error correcting codes}, or \emph{$(k_2,k_1)$-limited magnitude single error correcting codes} when $t=1$. If moreover, $\mathbb{Z}^n=\bigsqcup_{\mathbf{x}\in C}(\mathbf{x}+\mathcal{B}(n,t,k_2,k_1))$, we say that the code $C$ is \emph{perfect}.

Limited magnitude error correcting codes are apparently first investigated by Levenshtein and Vinck in \cite{levenshtein1993perfect}. What they called \emph{$k$-shift codes} are codes corresponding to the case that $k_1=k_2=k, t=1$. Since then, such codes are studied in \cite{munemasa1995perfect,tamm2002splittings,tamm2005perfect,szabo2011integer}.

We assume that $ C $ is a \emph{sublattice} of $ \mathbb{Z}^n $ in this paper. In this case, we also say $C$ is a \emph{linear} code correcting $(k_2,k_1)$-limited-magnitude errors. If $t=1$, consider the canonical projection
\[
	\mathbb{Z}^n \to \mathbb{Z}^n / C =: G,
\]
which maps a vector $ \mathbf{x} $ to its coset $ \mathbf{x} + C $.

For an element $ s $ in an abelian group $ (G, +) $, define the group homomorphism $ \phi_s : \mathbb{Z} \to G $ by $ \phi_s(1) = s $. By abuse of notation, we write $ \lambda s $ to denote $ \phi_s(\lambda) $. Let $ s_i $ denote the image of the standard basis vector $\mathbf{e}_i=(\underbrace{0,\ldots,0}_{i-1},1,0,\ldots,0)$  under this projection. Then $ C $ is a limited magnitude single error correcting code if and only if the elements $ \lambda s_i $ are distinct and nonzero for all $ \lambda \in [-k_1, k_2]^* := [-k_1, k_2] \setminus \{0\} $.

From now on, we assume that $ G = \mathbb{Z}_N $ is a cyclic group. If the elements $ \lambda s_i $, where $ \lambda \in [-k_1, k_2]^* $, are all distinct and nonzero, we say that the set $ B = \{s_i : i = 1, \ldots, n\} $ is a \emph{$ B[-k_1, k_2](N) $} splitter set
(For the case where $ G = \mathbb{Z}_q^r $ with $ r > 1 $, we refer to \cite{klove2014linear}, where the counterparts of $ B[-k_1, k_2](N) $ splitter sets are termed ``$ B[r, -k_1, k_2](q) $ sets''.). It would be straightforward to construct $(k_2,k_1)$-limited magnitude single error correcting codes of block length $n$ from splitter sets of size $n$ if existed.

Clearly, the size of such a set satisfies
\[
	|B| \le \frac{N - 1}{k_1 + k_2}.
\]
We will always assume that $B$ is nontrivial, i.e. $|B|>0$, so we must require that $k_1+k_2<N$. If the splitter set $ |B| $ attains the maximum value $ \left\lfloor \frac{N - 1}{k_1 + k_2} \right\rfloor $, we call it \emph{perfect} when $ (k_1 + k_2) \mid (N - 1) $, and \emph{quasi-perfect} when $ (k_1 + k_2) \nmid (N - 1) $.\footnote{Note that in \cite{klove2011systematic}, the term ``quasi-perfect'' may also refer to perfect splitter sets. In this paper, we maintain a distinction between the two concepts.}
A $ B[-k_1, k_2](N) $ splitter set is called \emph{nonsingular} if $ \gcd(N, k_1! \cdot k_2!) = 1 $, and \emph{singular} otherwise. 
It is easy to see that a limited magnitude single error correcting code is perfect if and only if its associated splitter set is perfect; if for some parameters $k_1,k_2,N$ we cannot have a perfect $B[-k_1,k_2](N)$ splitter set, quasi-perfect splitter sets, if they exist, provide codes that are second-best.

Our main results in this paper can be summarized as follows:
\subsection{Quasi-perfect splitter sets}
\subsubsection{A general result}
It is known that if $\gcd(m,k!)=1$, we have the following construction for a quasi-perfect $B[0,k](km)$ splitter set:
\begin{constr}[{\cite[Theorem IV.1]{ye2020some}}]
	\label{constr-quasi-1-k}
	Let $k,m$ be positive integers such that $\gcd(m,k!)=1$. Let $a\equiv (-k)^{-1}\pmod{m}$. Then
	$$B=\{ik+1 : i\in[0,m-1]\text{ and } i\ne a\}$$
	is a quasi-perfect $B[0,k](km)$ set.
\end{constr}
Hence, the case that $\gcd(m,k!)\ne 1$ would be worth considering. In our previous work, we proved that
\begin{prop}[{\cite[Proposition VII.1]{yuan2025existence}}]
	If $m>k$ and $k$ divides $m$, then there is no quasi-perfect $B[0,k](km)$ set.
\end{prop}
In this work, we will prove the following
\begin{prop}
	\label{quasi-ktm}
	If $t\in [2,k-1]$ is a prime, then there is no quasi-perfect $B[0,k](ktm)$ set for $m\ge 3$.
\end{prop}
\subsubsection{On quasi-perfect $B[0,3](n)$ splitter sets}
We obtain a necessary and sufficient condition for the existence of quasi-perfect $B[0,3](n)$ splitter sets:
\begin{prop}
	Let $k, l \ge 1$ and $r$ be positive integers such that $\gcd(r,6)=1$. The following statements hold:
	\begin{enumerate}
		\item For $l \ge 1$, a quasi-perfect $B[0,3](2^lr)$ splitter set can exist if and only if $l = 2$ and a quasi-perfect $B[0,3](r)$ splitter set exists.

		\item A quasi-perfect $B[0,3](3^lr)$ splitter set exists if and only if either
		      \begin{itemize}
			      \item $l=1$, or
			      \item $l=2$ and $r=1$.
		      \end{itemize}

		\item A quasi-perfect $B[0,3](2^k3^lr)$ splitter set exists if and only if $k = l = r = 1$.

	\end{enumerate}
	\label{thm:03}
\end{prop}
\subsubsection{On $B[-4,4](2p)$ splitter sets, where $p$ is a prime} We obtain the following proposition.
\begin{prop}
\label{prop-pm4}
Quasi-perfect $B[-4,4](2p)$ splitter sets exist for odd prime $p$ if and only if $p\equiv 1\pmod 4$ and $v_2(\ind_g(2))=v_2(\ind_g(3))<v_2(\frac{p-1}2)$, where $g$ is an odd primitive root modulo $p$, or $p=7$.
\end{prop}
\subsection{Perfect $B[0,6](q)$ splitter sets}

Preceding works have established many criteria for the existence or nonexistence of perfect splitter sets, see \cite{klove2011some, klove2012codes, zhang2017splitter, zhang2018nonexistence, ye2020some, yuan2025existence,yuan2020existence}. However, the existence condition for nonsingular perfect $B[0,6](q)$ splitter sets with $q$ prime is unknown. In this paper, we give the following necessary and sufficient condition for the existence of nonsingular perfect $B[0,6](q)$ splitter sets:
\begin{thm}
	Let $q\equiv 1\pmod 6$ be a prime. Fix a primitive root $g$ modulo $q$, and let $\alpha=\ind_g(2),\beta=\ind_g(3),\gamma=\ind_g(5)$. Then there exists a perfect $B[0,6](q)$ splitter set if and only if $q$ satisfies one of the following conditions:
	\begin{enumerate}
		\item $v_2(\beta)<v_2(q-1)$,
		      $v_3(\alpha)<v_3(q-1)$, $v_3(\alpha)<v_3(\beta)$, $v_2(\beta)<v_2(\gamma-2\alpha-\beta)$,
		      $v_3(\alpha)<v_3(\gamma-2\alpha-\beta)$;
		\item  $v_2(\gamma)<v_2(q-1)$, $v_3(\alpha)<v_3(q-1)$,
		      $v_3(\alpha)<v_3(\gamma)$,
		      $v_2(\gamma)<v_2(\alpha+\gamma-\beta)$, $v_3(\alpha)<v_3(\alpha+\gamma-\beta)$;
		\item  $v_2(\alpha+\beta)<v_2(q-1)$, $v_3(\alpha)<v_3(q-1)$,
		      $v_3(\alpha)<v_3(\alpha+\beta)$, $v_3(\alpha)<v_3(\gamma-2\alpha-\beta)$, $v_2(\alpha+\beta)<v_2(\gamma-2\alpha-\beta)$;
		\item  $v_3(\alpha)<v_3(q-1)$, $v_2(\beta)<v_2(q-1)$,
		      $v_2(\beta)<v_2(\alpha)$,
		      $v_2(\beta)<v_2(\gamma-\beta-2\alpha)$, $v_3(\alpha)<v_3(\gamma-\beta-2\alpha)$;
		\item $v_2(\alpha)<v_2(q-1)$, $v_2(\alpha)<v_2(\beta)$, $ v_3(\beta)<v_3(q-1)$,
$v_2(\alpha)<v_2(\gamma-3\alpha)$,
$v_3(\beta)<v_3(\gamma-3\alpha)$,
$v_3(\beta)<v_3(2\alpha+\beta)$.%
	\end{enumerate}
	\label{thm:06}
\end{thm}

\subsection{Perfect limited magnitude burst correcting codes}
Wei and Schwartz \cite{wei2022perfect} considered the limited magnitude version of burst correcting. In this setting, errors are assumed to be confined to an interval of a certain length. They distinguished between cyclic and noncyclic bursts, where \emph{cyclic} bursts admit a sequence of errors that starts at the end of a codeword and ends at the beginning of the codeword,
while \emph{noncyclic} ones do not. The error balls of these two types of error are
$$
	\begin{aligned}
		\mathcal{E}^\circ(n,b,k_2,k_1)= &
		\{(e_0,e_1,\ldots,e_{n-1})\in [-k_1,k_2]^n:                                                           \\
		                                & \text{there is an }i\in\mathbb{Z}_n\text{ such that }e_\ell=0     \\
		                                & \text{for all }\ell\in\mathbb{Z}_n\setminus\{i,i+1,\ldots,i+b-1\}
		\}
	\end{aligned}
$$
(for cyclic bursts) and
$$
	\begin{aligned}
		\mathcal{E}(n,b,k_2,k_1)= &
		\{(e_1,e_2,\ldots,e_{n})\in [-k_1,k_2]^n:                                    \\
		                          & \text{there is an }i\in[1,n]\text{ such that }
		e_\ell=0                                                                   \\&\text{for all }\ell\in[1,n]\setminus[i,\min\{n,i+b-1\}]\}
	\end{aligned}
$$
(for noncyclic bursts), respectively.
They constructed infinite families of perfect 2-burst correcting codes for $(1,0)$-limited-magnitude errors, as well as 
perfect cyclic $\le 3$-burst correcting codes for $(1,1)$- and $(1,0)$-limited-magnitude errors, by constructing appropriate group splittings. The problem of finding perfect codes with longer bursts or of larger magnitude is still open.

In this paper, we will give a general framework for constructing perfect limited magnitude cyclic burst correcting codes, and show the existence of infinitely many such codes for $k_1+k_2\le 4$ and arbitrary burst length $b$. We also give sufficient existence conditions for general parameters $k_1, k_2$. %

\subsection{Maximal size of splitter sets}
Another research topic, started in \cite{klove2011systematic}, is concerned with how large a $B[-k_1,k_2](N)$ splitter set can be for different $N$. Define $M(k_1,k_2;N)$ to be the maximal size of a $B[-k_1,k_2](N)$ splitter set. 
Many results about $M(k_1,k_2;N)$ can be found in \cite{klove2011some,klove2011systematic,xie2020asymmetric, xie2019further,yari2013some,klove2012codes}.
In \cite{roche2018packing}, the general problem of determining the maximal size of an $A$-packing set, denoted by $\nu(A)$, i.e.
$$
	\nu(A):=\max\{|B|:B\subset G, |A\cdot B|=|A||B|\},
$$
was considered.
Graph-theoretic methods may be applied to study $\nu(A)$.
 Notice that $|A\cdot B|=|A||B|$ is equivalent to
$$
	A\cdot A^{-1}\cap B\cdot B^{-1}=\{1\}.
$$
We construct a graph $\Gamma=(V,E)$ with vertex set $V=G$, and edge set $E=\{\{x,y\}:x,y\in G,x\ne y, xy^{-1}\in A\cdot A^{-1}\}$.
It is the \emph{Cayley graph} $\Cay(G;A\cdot A^{-1})$ defined by $A\cdot A^{-1}$, see \cite{ye2020some}.
Note that $\Gamma$ is an $(|A\cdot A^{-1}|-1)$-regular graph.
It is easy to verify that $B$ is an $A$-packing set if and only if $B$ is an independent set in $\Gamma$. By Turán's theorem (as considered in \cite{roche2016packing}), we get
$$
	\nu(A)\ge \left\lceil\frac{|G|}{|A\cdot A^{-1}|}\right\rceil.
$$
A similar result obtained by Brooks' theorem can be found in \cite{ye2020some}.
In this paper, we will derive a lower bound for $M(0,3;q)$ from this graph-theoretic perspective, where $q$ is a prime.

\section{Preliminaries}
\subsection{Group factorization}
A multiset is a set whose elements have multiplicities. For a finite abelian group $(G,+)$ and multisets $A_1,\ldots,A_n$ over $G$, we define
$$
	A_1+A_2+\cdots+A_n=\{a_1+a_2+\cdots+a_n:a_i\in A_i\}
$$
as the \emph{sum} of $A_1,A_2,\ldots,A_n$. For example, $\{0,0\}+\{1,2\}=\{1,1,2,2\}$ over $\mathbb{Z}_3$.
If $A_1,\ldots, A_n$ are ordinary sets, sums $a_1+\cdots+a_n, a_1\in A_1,\ldots,a_n\in A_n$ are distinct, we say that
$A_1+A_2+\cdots+A_n$ is a \emph{direct sum}. If  moreover $G=A_1+\cdots+A_n$, we say $G=A_1+\cdots+A_n$ is a \emph{factorization} of $G$, and every $A_i$ is called its \emph{direct factor}. Notice that if $G=A_1+\cdots+A_n$ as a multiset, where $A_1,\ldots,A_n$ are multisets over $G$, then $G=A_1+\cdots+A_n$ is automatically a factorization of $G$. If $n=2$ and $G=A_1+A_2$ is a factorization, $A_2$ will be called a \emph{complementary factor} of $A_1$.

If $A_1+A_2$ is a direct sum and $G\setminus\{f\}=A_1+A_2$ for some $f\in G$, we say that $G\setminus\{f\}=A_1+A_2$ is a \emph{near-factorization}\footnote{In the original definition of near-factorization \cite{caen1990nearfactor}, $f$ was required to be $0$.} of $G$. 
\begin{lemma}[{\cite[Corollary 2.7]{buratti1996packing}}]
	Let $\{0,x\}\subset G=\mathbb{Z}_{2k+1}$. There exists an $A\subset G$ and an $f\in G$ such that
	$$G\setminus\{f\}=A+\{0,x\}$$
is a near-factorization if and only if $\gcd(x,2k+1)=1$.
\label{lem-2-near}
\end{lemma}
For a(n) (ordinary) subset $A$ of an abelian group $G$, let $\Delta A=\{i-j:i,j\in A, i\ne j\}$ (which is an ordinary set in this paper). It is clear that for $x\ne 0$, $A\cap(A+x)=\varnothing$ if and only if $x\notin\Delta A$. If $A+x=A$, we say that $x$ is a \emph{period} of $A$.

\begin{prop}
	Let $n=2k+1$, $k\ge 1$, $x\in \mathbb{Z}_n$, $x\ne 0$, $A\subset\mathbb{Z}_n$ is an ordinary set of size $k$, such that $\mathbb{Z}_n\setminus\{a\}=\{0,x\}+A$ is a near-factorization for some $a$, or equivalently,
	$$
		A\cap (A+x)=\varnothing,
	$$
	then for any $x'\in\mathbb{Z}_n$, $x'\ne \pm x$, $A\cap(A+x')=\varnothing$ cannot hold.
	\label{prop-2-uniqueness}
\end{prop}
\begin{proof}
	The cases $k=1,2$ are checked directly. Hence assume $k\ge3$. There is no loss of generality in the assumption that $x=1$, since by Lemma \ref{lem-2-near}, $\gcd(x,n)=1$, and we may replace $A$ with $A/x$. Now we have
	$$
		A\sqcup (A+1)\sqcup\{a\}=\mathbb{Z}_n
	$$
	for some $a\in\mathbb{Z}_n$. Consider $A+2$.
	We have $(A+1)\cap(A+2)=\varnothing$ since $A\cap (A+1)=\varnothing$. Then $(A+2)\subset A\sqcup\{a\}$.
	Notice that $|A+2|=|A|=k$. If $A+2=A$, then $2$ is  a period of $A$, hence $A=A+2(k+1)=A+1$, a contradiction.
	So $A+2$ must contain $a$, and there is a $b\in\mathbb{Z}_n$, $b\ne a$, such that
	$$
		(A+2)\sqcup\{b\}=A\sqcup\{a\}.
	$$
	From this we see that
	\begin{enumerate}
		\item if $c\in A$, $c\ne b$, then $c\in (A+2)$ and hence $c-2\in A$;
		\item similarly, if $c\in A$ and $c+2\ne a$, then $c+2\in A$.
	\end{enumerate}
	Now pick $c\in A$, $c\ne b,a-2$. This is always possible for $k \ge 3$. Then $c\pm 2\in A$.
	If $c+2i\in A$ for $i=1,2,\ldots$, then $c+2\mathbb{Z}_n=\mathbb{Z}_n\subset A$, a contradiction. So there exists $i_0\ge 1$ such that $c+2i_0\in A$ and $c+2(i_0+1)\notin A$, which implies that $c+2i_0=a-2$.
	Similarly there exists $j_0\ge 1$ such that $c-2j_0=b$.
	Since $c$ is chosen arbitrarily,
	$A$ must be
	$\{b,b+2,b+4,\ldots,b+2k-2\}$ where $b+2k=a$.
	Then $\Delta A=\{\pm 2,\pm4,\ldots,\pm2(k-1)\}=
		\{2,4,\ldots,2(k-1),2k-1,2k-3,\ldots,3\}=\mathbb{Z}_n\setminus\{0,\pm1\}$.
	So if $x'\in\mathbb{Z}_n$ such that $A\cap (A+x')=\varnothing$, then $x'=\pm1$.
\end{proof}

\begin{lemma}[{\cite[Theorem 2.9]{buratti1996packing}}]
	Let $\{0,x,y\}\subset G=\mathbb{Z}_{3k+1}$. There exists an $A\subset G$ and an $f\in G$ such that
	$$G\setminus\{f\}=A+\{0,x,y\}$$
is a near-factorization if and only if $\gcd(x,y,3k+1)=1$ and one of the following
	identities holds in $G$: $x+y=0$, $2x=y$, $x=2y$.
\label{lem-3-near}
\end{lemma}
\subsection{Mask polynomial}
With a finite multiset of nonnegative integers $A$, we associate a \emph{mask polynomial}:
$$
	f_A(x)=\sum_{i\in A} x^i.
$$
For example, $f_{\{0,1\}}=1+x$, $f_{\{0,1,1\}}=1+2x$. Mask polynomials naturally relate multisets to elements in $\mathbb{Z}[x]$. Under this correspondence, the multiset union (denoted by $A \uplus B$) corresponds to the addition of their associated mask polynomials. Furthermore, the sum of multisets corresponds to the multiplication of these polynomials. To be specific,
\begin{lemma}
	Let $A,B$ be two finite multisets of nonnegative integers, then $$
	f_A(x)+f_B(x)=f_{A\uplus B}(x)
	$$ and $$
		f_A(x)f_B(x)=f_{A+B}(x).$$
\end{lemma}

These definitions and properties extend naturally to multisets of elements in $\mathbb{Z}_n$, where the mask polynomial $f_A(x)$ is viewed as an element of the quotient ring $\mathbb{Z}[x]/(x^n - 1)$. The correspondence between multiset addition in $\mathbb{Z}_n$ and polynomial multiplication remains valid under this setting:
\begin{lemma}
\label{mask-factor}
Let $A,B$ be two finite multisets of elements in $\mathbb{Z}_n$, then $$
	f_A(x)+f_B(x)\equiv f_{A\uplus B}(x)\pmod{(x^n-1)}
	$$ and $$
		f_A(x)f_B(x)\equiv f_{A+B}(x)\pmod{(x^n-1)}.$$
\end{lemma}
Throughout this paper, when we explicitly see $A$ as a multiset, we denote by $|A|$ the \emph{cardinality} or \emph{size} of $A$ as a multiset (i.e., the sum of the multiplicities of its elements). For instance, $|\{0, 0, 1\}| = 3$, and under this convention, we always have that $f_A(1) = |A|$.
\subsection{Group splitting}
A concept called ``group splittings'' is connected with perfect splitter sets.
\begin{dfn}
	Let $M$ be a finite set of nonzero integers and $(G,+)$ a finite abelian group.
	We say that $G$ has a splitting $G\setminus\{0\}=MS$ if there exists a subset $S\subset G$ such that every nonzero element of $G$ can be uniquely represented as $ms, m\in M, s\in S$. In this case, we say $M$ is the \emph{multiplier set} of the splitting and $S$ its \emph{splitter set}\footnote{It was named \emph{splitting set} in {\cite[10.3]{szabo2009factoring}}. Here we follow {\cite[12.1]{etzion2022perfect}} in consistency with our definition of splitter sets.}, and we say that $M$ splits $G$.
\end{dfn}

For group splittings, interested readers can refer to \cite{hickerson1983splittings, stein1967factoring,stein1972symmetric,
	szabo2009factoring}. We now generalize this concept to include nonperfect splitter sets:
\begin{dfn}
	Let $M$ be a finite set of nonzero integers and $(G,+)$ a finite abelian group.
	We say that $M$ \emph{partially splits} $G$ with a splitter set $S$ if $|MS|=|M||S|$ and $0\notin MS$. If $M$ partially splits $G$, we say that $G$ has a partial splitting $G\setminus\{0\}\supset MS$ and $M$ is the \emph{multiplier set} of this partial splitting and $S$ its \emph{splitter set}.
\end{dfn}

\subsection{Number theoretic tools}

The $n$-th cyclotomic polynomial $\Phi_n(x)$
is defined to be the monic minimal polynomial of an $n$-th primitive root of unity over $\mathbb{Q}$, and hence is irreducible in $\mathbb{Q}[x]$. In fact, $\Phi_n(x)\in\mathbb{Z}[x]$. For $n\ne m$, $\gcd(\Phi_n(x),\Phi_m(x))=1$ since they have no common zeros in $\mathbb{C}$.

\begin{lemma}[{\cite[\S 46]{nagell2021introduction}}] We have
	$$\Phi_n(1)=\begin{cases}p,&\text{if $n$ is a power of a prime $p$,}\\1,&\text{otherwise.}\end{cases}$$
\end{lemma}

For every odd prime $p$ with a primitive root $g$,
let $\ind_g(x)$ be the least integer $e$ such that $x\equiv g^e\pmod{p}$. If $g$ is odd, then $g$ is also a primitive root modulo $2p$. The multiplicative group $\mathbb{Z}_{2p}^\times$ consists of all odd numbers in $\mathbb{Z}_{2p}$ except $p$, and for $x\in\mathbb{Z}_{2p}^\times$, we can prove that $\ind_g(x\bmod p)$ is also the least integer $e$ that $x\equiv g^e\pmod{2p}$. So in this case, we also write $\ind_g(x)$ for such an integer $e$.

In some proofs we will use results from algebraic number theory. We will denote by $\mathcal{O}_K$ the ring of integers of a number field $K$. 

\section{A new nonexistence result on quasi-perfect splitter sets}
In this section, we prove Proposition \ref{quasi-ktm}.
\begin{proof}
	Suppose $B$ is a quasi-perfect $B[0,k](ktm)$ set, then
	$|B|=tm-1$.
	Notice that there are $tm-1$ multiples of $k$ in $\mathbb{Z}_{ktm}\setminus\{0\}$ and $|kB|=|B|=tm-1$, thus $kB=\{k, 2k,\ldots, k(tm-1)\}$ is exactly the set of all multiples of $k$ in $\mathbb{Z}_{ktm}\setminus\{0\}$.
	Moreover, since the solutions to the equation
	$$kx\equiv ki\pmod{ktm}$$
	in $[0,ktm-1]$ are of the form $i+ytm$, where $y$ is some integer in $[0,k-1]$, so the set $B$ must be of the form
	$\{i+f(i)tm : i\in [1,tm-1]\}$ where $f$ is some function from $[1,tm-1]$ to $[0,k-1]$, and elements of $B$ are congruent to $1,2,\ldots,tm-1$ modulo $tm$.

	Now consider the multiples of $t$. There are $km-1$ multiples of $t$ in $\mathbb{Z}_{ktm}\setminus\{0\}$. Consider the set $iB, i\in [1,k]$. If $t\mid i$, then $tm-1$ distinct elements of $iB$ are all multiples of $t$; otherwise, there are at least $m-1$ multiples of $t$ in $iB$. By the discussion above, we can calculate there are at least
	$$
		\begin{aligned}
			 & \left\lfloor \frac{k}t\right\rfloor(tm-1)+\left(k-\left\lfloor \frac{k}t\right\rfloor\right)(m-1)
			\\=&k(m-1)+\left\lfloor \frac{k}t\right\rfloor((t-1)m)
		\end{aligned}$$
	multiples of $t$ in $\bigsqcup_{i=1}^k iB$. Compare it to the total number $km-1$, we deduce that
	$$
		\begin{aligned}
			m & \le\frac{k-1}{(t-1)\left\lfloor \frac{k}t\right\rfloor}
			\le\frac{t\left\lfloor \frac{k}t\right\rfloor+t-2}{(t-1)\left\lfloor \frac{k}t\right\rfloor}                          \\
			  & =1+\frac{1}{t-1}+\frac{1}{\left\lfloor \frac{k}t\right\rfloor}-\frac{1}{(t-1)\left\lfloor \frac{k}t\right\rfloor}
			\le 3-\frac{1}{(t-1)\left\lfloor \frac{k}t\right\rfloor}<3.
		\end{aligned}
	$$
\end{proof}

\section{Existence of quasi-perfect $B[0,3](n)$ splitter sets}
In this section, we discuss quasi-perfect $B[0,3](n)$ sets.
\subsection{Existence of nonsingular quasi-perfect $B[0,3](n)$ splitter sets}
We first consider \emph{nonsingular} quasi-perfect $B[0,3](n)$ sets for composite $n$. First recall that
\begin{thm}[{\cite[Theorem 5]{yari2013some}}]
	Let $B_1$ be a $B[-k_1,k_2](q_1)$ set and
	$B_2$ be a $B[-k_1,k_2](q_2)$ set, where
	$\gcd(k_1!k_2!,q_2)=1$.
	Then
	$$
		\{ c+rq_1 : c\in B_1, r\in [0,q_2-1]\}\cup\{q_1c:c\in B_2\}
	$$
	is a $B[-k_1,k_2](q_1q_2)$ set of size $q_2|B_1|+|B_2|$.
	Especially, if $B_1,B_2$ are perfect, then there exists a perfect $B[-k_1,k_2](q_1q_2)$ set.
	\label{thm-q1q2}
\end{thm}
If $B_1$ is perfect, that is, $|B_1|=\frac{q_1-1}{k_1+k_2}$, and $B_2$ is quasi-perfect,
that is $|B_2|=\frac{q_2-1-t}{k_1+k_2}$, where
$t\in [1,k_1+k_2-1]$.
Then there exists a $B[-k_1,k_2](q_1q_2)$
set of size $q_2\frac{q_1-1}{k_1+k_2}+\frac{q_2-1-t}{k_1+k_2}=\frac{q_1q_2-1-t}{k_1+k_2}$, which is a quasi-perfect $B[-k_1,k_2](q_1q_2)$ set.
So we have the following
\begin{prop}
	Suppose that $\gcd(n,m)=1$, $\gcd(6,m)=1$ and there exists a
	\emph{perfect}
	$B[0,3](n)$ splitter set
	and a \emph{quasi-perfect}
	$B[0,3](m)$ splitter set. Then there exists a
	\emph{quasi-perfect} $B[0,3](nm)$ splitter set.
\end{prop}
Now we study in order that a quasi-perfect $B[0,3](n)$ set exists, where $n$ is a composite coprime to $6$,
what conditions $n$ should satisfy. Recall that for group splittings, we have the following result:
\begin{thm}[{\cite[Theorem 3]{hamaker1974splitting}}]
	\label{thm-nonsingular}
	Let $G$ be a finite abelian group and
	$$
		\{0\}\to A\to G\to B\to \{0\}
	$$
	an exact sequence.
	Assume that $M$ splits $G$.
	If each element of $M$ is relatively prime
	to $|B|$, then $S$ splits $A$.
	If each element of $M$ is relatively prime to
	$|A|$, then $S$ splits $B$.
\end{thm}
The following is a generalization of the theorem above.
\begin{prop}
	\label{prop-part-spl}
	Suppose a finite abelian group $G$ has a partial splitting
	$$
		G\setminus\{0\}\supset MS
	$$
	and $A$ is a subgroup of $G$ whose index in $G$ is coprime to every element of $M$, then
	$A\cap MS=M(A\cap S)$ and hence $M$ partially splits $A$ with $A\cap S$.
\end{prop}
\begin{proof}
	Let $\beta$ be the canonical map $G\to G/A$.
	Since $A$ is a group, we always have
	$M(A\cap S)\subset A\cap MS$.
	Now suppose $ms\in A\cap MS$,
	then $m\beta(s)=\beta(ms)=0$ and $\gcd(m,|G/A|)=1$
	implies that $\beta(s)=0$ and hence $s\in A$.
	So $M(A\cap S)=A\cap MS$.
\end{proof}
Now suppose that $n$ is a composite number, $\gcd(6,n)=1$, and that there exists a quasi-perfect $B[0,3](n)$ set $B$, that is,
$$
	B\sqcup 2B\sqcup 3B\sqcup\{0,f\}=\mathbb{Z}_n
$$
for some $f\ne 0$.
Then, for any subgroup $G\subset\mathbb{Z}_n$
such that $f\notin G$, we have $G\setminus\{0\}\subset [1,3]\cdot B$, so $G\cap B$ is a perfect $B[0,3](|G|)$ splitter set, and hence
$|G|\equiv 1\pmod 3$;
for any subgroup $G\subset\mathbb{Z}_n$ such that
$f\in G$, we have $G\setminus\{0,f\}\subset [1,3]\cdot B$, so $G\cap B$ is a quasi-perfect $B[0,3](|G|)$ splitter set, which implies that $|G|\equiv 2\pmod 3$.

So, for every divisor $m$ of $n$,
if $m\equiv 1\pmod 3$, then
there must exist a perfect $B[0,3](m)$ set;
if $m\equiv 2\pmod 3$, then
there must exist a quasi-perfect $B[0,3](m)$ set.
It implies that every divisor $m$ of $n$ of the form
$3k+1$ cannot have any prime divisor of the form $3k+2$ (by, for example, \cite[Theorem 19]{schwartz2014nonexistence}).
So there is only one prime divisor $q$ of $n$ of the form $3k+2$, which must admit a quasi-perfect $B[0,3](q)$ set
and the highest power of which that divides $n$ must be $q^1$.

Now we determine all the primes $q$ that admit a perfect $B[0,3](q)$ splitter set. As in our previous work {\cite[proof of Corollary III.15]{yuan2025existence}}, we can prove that the necessary and sufficient condition for the existence of a perfect $B[0,3](q)$ set
is $v_3(\ind_g(2))=v_3(\ind_g(3))=v_3(\ind_g(3)-\ind_g(2))<v_3(q-1)$, where $g$ is a primitive root modulo $q$. We can also show that it is equivalent to $2\notin\langle 6,8\rangle\subset\mathbb{Z}_q^\times$.

Next, we determine all the primes $q$ that admit a quasi-perfect $B[0,3](q)$ splitter set. We suppose that $q=3k+2$ is such a prime. Clearly, the size of the corresponding quasi-perfect splitter set is $k$. As we did in the perfect case, we use the primitive root to turn the problem into an additive one. Fix a primitive root $g$ modulo $q$. If $B$ is a quasi-perfect $B[0,3](q)$ splitter set, then
let $A=\{\ind_g(x):x\in B\}$, thus
$$A\sqcup (A+\ind_g(2))\sqcup (A+\ind_g(3))\sqcup \{f\}=\mathbb{Z}_{q-1}$$
for some $f$ in $\mathbb{Z}_{q-1}$. By Lemma \ref{lem-3-near}, we have a quasi-perfect $B[0,3](q)$ set
if and only if $\gcd(\ind_g(2),\ind_g(3),q-1)=1$ and one of the following equations holds modulo $q-1$:
$\ind_g(2)+\ind_g(3)=0$,
$\ind_g(2)=2\ind_g(3)$, $\ind_g(3)=2\ind_g(2)$.
The first equation is equivalent to saying that $6\equiv1\pmod{q}$, which holds if and only if $q=5$.
The second equation is equivalent to saying that $9\equiv2\pmod{q}$, which holds if and only if $q=7$, but $7\not\equiv 2\pmod 3$.
The third equation is equivalent to saying that $3\equiv4\pmod{q}$, which never holds.

Hence, the only prime $q$ that admits a quasi-perfect $B[0,3](q)$ splitter set is $5$. We have the following
\begin{prop}
	Nonsingular quasi-perfect $B[0,3](n)$ sets exist if and only if
	$5\mid n$, and every prime divisor $p$ of $n/5$ admits a perfect $B[0,3](p)$ splitter set.
\end{prop}
\subsection{Results on singular splitter sets}
Now we consider \emph{singular} quasi-perfect $B[0,3](n)$ splitter sets. There are three cases:
\begin{enumerate}
	\item $2\mid n$, $3\nmid n$;
	\item $2\nmid n$, $3\mid n$;
	\item $6\mid n$.
\end{enumerate}

In the first case, we assume that $n=2^l r$, $l\ge 1$, $ \gcd(r,6)=1$.

\begin{thm}[{\cite[Corollary 1]{klove2011some}}]
	\label{thm-3-2r}
	Assume $n=2^kr$ with $\gcd(r,6)=1$ and $k\ge 2$. Then a quasi-perfect $B[0,3](n)$ set exists if and only if
	$k=2$ or $3$ and a quasi-perfect $B[0,3](2^{k-2}r)$ set exists.
\end{thm}

Theorem \ref{thm-3-2r} has settled the cases $l\ge 2$. Now we consider
$B[0,3](2r)$ sets.
If there exists a quasi-perfect $B[0,3](2r)$ set $B$, such that
$$
	B\sqcup 2B\sqcup 3B\sqcup\{0,f\}=\mathbb{Z}_{2r},
$$
then $2r\equiv 2\pmod 3$, and hence
$r\equiv 1\pmod 3$.
Suppose $q$ is an odd divisor of $r$, by Proposition \ref{prop-part-spl},
the subgroup $A$ of $(\mathbb{Z}_{2r},+)$ of index $q$ is partially split by $[1,3]$, and $A\setminus\{0\}=[1,3]\cdot(A\cap B)$ if $f\notin A$, $A\setminus\{0,f\}=[1,3]\cdot(A\cap B)$ if $f\in A$. From this, we obtain that for every divisor $q$ of $r$ such that $q\equiv 2\pmod3$, the subgroup of index $q$ is split by $[1,3]$ and does not contain $f$; for every divisor $q$ of $r$ such that $q\equiv 1\pmod 3$, the subgroup of index $q$ is partially split by $[1,3]$ and contains $f$. Hence we see that the $2$-subgroup of $\mathbb{Z}_{2r}$ is precisely $\{0,f\}$, and $f$ is the only element of order $2$ in $\mathbb{Z}_{2r}$, i.e. $f=r$.
However, every subgroup of odd index in $\mathbb{Z}_{2r}$ must contain $f$, and thus its order must be congruent to $2$ modulo $3$. Hence, no divisor of $r$ is congruent to $2$ modulo $3$.

Now, suppose that $p\equiv1\pmod3$ is a prime divisor of $r$. Then there must exist a quasi-perfect
$B[0,3](2p)$ set, $B'$.
Let $B_0=\{i\in B':2\mid i\}$,
$B_1=\{i\in B':2\nmid i\}$. Then we have
$$
	\mathbb{Z}_{2p}=\{0,p\}\sqcup (B_0\sqcup 2B_0\sqcup 3B_0\sqcup 2B_1)\sqcup(B_1\sqcup 3B_1).
$$
Notice that elements in $B_0\sqcup 2B_0\sqcup 3B_0\sqcup 2B_1$ are even, and those in $B_1\sqcup 3B_1$ are odd, hence $\{0\}\sqcup B_0\sqcup 2B_0\sqcup 3B_0\sqcup 2B_1=2\mathbb{Z}_p$ which is the subgroup of index $2$ in $\mathbb{Z}_{2p}$, and $B_1\sqcup 3B_1=\mathbb{Z}_{2p}^\times$.

Now let $C_0=\{i\in[0,p-1]:2i\in B_0\}\subset\mathbb{Z}_p$,
$C_1=B_1\pmod{p}=\{i\in[0,p-1]:i\in B_1 \text{ or }i+p\in B_1\}\subset\mathbb{Z}_p$.
From $\{0\}\sqcup B_0\sqcup 2B_0\sqcup 3B_0\sqcup 2B_1=2\mathbb{Z}_p$ we have
$$
	\{0\}\sqcup C_0\sqcup 2C_0\sqcup 3C_0\sqcup C_1=\mathbb{Z}_p.
$$
However, since $B_1\sqcup 3B_1=\mathbb{Z}_{2p}^\times$, by taking modulo $p$, we have
$$
\{0\}\sqcup C_1\sqcup 3C_1=\mathbb{Z}_p.
$$
Hence $3C_1=C_0\sqcup 2C_0\sqcup 3C_0$. Therefore $C_1=\frac13 C_0\sqcup\frac23 C_0\sqcup C_0$, and $C_0\subset C_1$, so $C_0=\varnothing$. Thus, $C_0\sqcup 2C_0\sqcup 3C_0=\varnothing$, and $C_1=\varnothing$, a contradiction. Hence, there does not exist quasi-perfect $B[0,3](2r)$ splitter sets.

In the second case, we suppose that $n=3^l r$, $l\ge 1$, $\gcd(r,6)=1$.

\begin{thm}[{\cite{klove2011some}}]
	For $n=3^lr$ and $l\ge 2$, there do not exist quasi-perfect $B[0,3](n)$ sets except for $l=2, r=1$.
	\label{thm-3lr}
\end{thm}

Theorem \ref{thm-3lr} has dealt with the case $l\ge 2$. If $l=1$, then Construction \ref{constr-quasi-1-k} gives
us a concrete quasi-perfect $B[0,3](3r)$ set.

Now we consider the last case. By Proposition \ref{quasi-ktm}, for $m\ge 3$, quasi-perfect $B[0,3](6m)$ splitter sets do not exist.
If $ m = 2 $, suppose there exists a quasi-perfect $ B[0,3](12) $ splitter set $ B $. As in the proof of Proposition \ref{quasi-ktm}, we can write
\[
	B = \{1 + 4f(1),\ 2 + 4f(2),\ 3 + 4f(3)\},
\]
with $ f\colon [1,3] \to [0,2] $. Since no element of $ B $ is divisible by 3, we require:
\begin{itemize}
	\item $ f(1) \neq 2 $ (so $ 1+4f(1) \not\equiv 0 \pmod{3} $),
	\item $ f(2) \neq 1 $ (so $ 2+4f(2) \not\equiv 0 \pmod{3} $),
	\item $ f(3) \neq 0 $ (so $ 3+4f(3) \not\equiv 0 \pmod{3} $).
\end{itemize}

Thus $ f(3) \in \{1,2\} $. Consider $ 2(3+4f(3)) \mod 12 $:
\begin{itemize}
	\item If $ f(3)=1 $, then $  2(3+4f(3))\equiv2 \pmod{12}$.
	\item If $ f(3)=2 $, then $  2(3+4f(3))\equiv10\pmod{12} $.
\end{itemize}

If $f(3)=1$, then $2+4f(2)\ne 2$ and $f(2)\ne 0$, so $f(2)=2$.
Since $2+8f(1)\ne 2$, $f(1)=1$. But then
$2(1+4f(1))=10=2+4f(2)$, a contradiction.

If $f(3)=2$, then $2(1+4f(1))\ne 10$ and $f(1)\ne 1$, so $f(1)=0$.
Since $2+4f(2)\ne 2(1+4f(1))=2$, $f(2)\ne 0$, so $f(2)=2$ and $2+4f(2)=10=2(3+4f(3))$, a contradiction.

Hence, no such splitter set exists.
The results of this subsection give Theorem \ref{thm:03}.

\section{On maximal $B[0,3](q)$ sets for $q$ prime}
In \cite{klove2011systematic}, a construction of $B[0,3](q)$ set of size approximately $q/6$ was given. In this section, we use a more careful argument on the structure of Cayley graphs to obtain a better lower bound for $M(0,3;q)$, where $q$ is a prime. We use a result in
\cite{fajtlowicz1984independence}:
\begin{thm}
	Let $G$ be a finite simple $n$-vertex graph of maximum degree $d$,
	of clique size $(k-1)$ (that is, is $K_k$-free) and independence $\alpha$.
	If $d\ge k$, then
	$$
		\frac{\alpha}n\ge\frac{2}{d+k}.
	$$
\end{thm}
\begin{lemma}
	If $q>7$ is a prime, then the following 6 elements
	$$
		2,3,2/3,1/2,1/3,3/2
	$$
	are distinct modulo $q$. Moreover, if $q>23$, then each of the following 6 elements
	$$
		4/3, 6, 9/2, 9, 4, 9/4
	$$
	does not belong to $\{2,3,2/3,1/2,1/3,3/2\}$ modulo $q$.
\end{lemma}
\begin{proof}
	Simple verification.
\end{proof}
\begin{thm}
	For every odd prime $q>23$, we have $M(0,3;q)\ge (q-1)/5$.
\end{thm}
\begin{proof}
	Let $S=[1,3]\cdot [1,3]^{-1}\setminus\{1\}=\{2,3,2/3,1/2,1/3,3/2\}$. We consider the Cayley graph $\Gamma=\Cay(G;S)$.

	If $q>7$, then elements of $S$ are distinct modulo $q$.
	Then $\Gamma$ is a 6-regular graph, and for every vertex, the induced subgraph on its 6 neighbors is a cycle $C_6$.
	In fact, the 6 neighbors of the vertex $i$ are
	$2i, 3i, 2i/3, i/2, i/3, 3i/2$ modulo $q$. Simple verification shows that there are edges $\{2i,3i\}, \{3i, 3i/2\}, \{3i/2, i/2\}, \{i/2, i/3\}, \{i/3, 2i/3\}, \{2i/3,2i\}$ between them. Moreover, $q>23$ ensures that there are no other edges between these vertices.
	Especially, $\Gamma$ has no $K_4$.
	So it has an independent set of size $\alpha(\Gamma)\ge 2(q-1)/(6+4)=(q-1)/5$.
\end{proof}
As an example, the Cayley graph $\Cay(\mathbb{Z}_{29}^\times; [1,3]\cdot [1,3]^{-1})$ and the neighborhood of $1$ in it are shown in Fig \ref{fig-z29} and Fig \ref{fig-z29-nbd}, respectively.
Using SageMath, we obtained exact values of $M(0,3;q)$ for small $q$, listed in Table \ref{tab-m03q}.
We notice that when $q>23$, the lower bound above is far from tight.
\begin{table*}
	\begin{tabular}{lllll}
		$q$ & $M(0,3;q)$ & A Maximal Splitter Set                                                            & Ratio $M(0,3;q)/(q-1)$ & \\
		11  & 2          & $\{10, 1\}$                                                                       & 0.2                    & \\
		13  & 3          & $\{10, 1, 6\}$                                                                    & 0.25                   & \\
		17  & 4          & $\{10, 6, 7, 11\}$                                                                & 0.25                   & \\
		19  & 5          & $\{10, 3, 4, 7, 17\}$                                                             & 0.277777777777778      & \\
		23  & 4          & $\{10, 1, 4, 14\}$                                                                & 0.181818181818182      & \\
		29  & 8          & $\{10, 2, 9, 12, 14, 16, 22, 23\}$                                                & 0.285714285714286      & \\
		31  & 8          & $\{10, 1, 4, 9, 14, 23, 25, 26\}$                                                 & 0.266666666666667      & \\
		37  & 12         & $\{10, 1, 6, 8, 11, 14, 23, 26, 27, 29, 31, 36\}$                                 & 0.333333333333333      & \\
		41  & 11         & $\{10, 40, 3, 4, 16, 18, 21, 26, 28, 33, 35\}$                                    & 0.275                  & \\
		43  & 12         & $\{10, 40, 2, 7, 9, 11, 12, 16, 23, 28, 29, 39\}$                                 & 0.285714285714286      & \\
		47  & 12         & $\{10, 40, 1, 4, 21, 22, 27, 28, 31, 35, 36, 45\}$                                & 0.260869565217391      & \\
		53  & 14         & $\{10, 40, 2, 8, 11, 17, 18, 26, 28, 37, 46, 48, 49, 50\}$                        & 0.269230769230769      & \\
		59  & 17         & $\{10, 1, 6, 7, 11, 13, 16, 19, 29, 34, 37, 41, 42, 45, 54, 55, 56\}$             & 0.293103448275862      & \\
		61  & 17         & $\{10, 2, 9, 11, 12, 13, 17, 19, 25, 28, 32, 41, 46, 53, 54, 55, 60\}$            & 0.283333333333333      & \\
		67  & 19         & $\{10, 40, 3, 4, 14, 17, 18, 19, 23, 24, 31, 33, 39, 41, 44, 47, 52, 63, 64\}$    & 0.287878787878788      & \\
		71  & 20         & $\{10, 2, 7, 9, 12, 13, 16, 23, 31, 33, 35, 37, 41, 43, 44, 45, 50, 60, 63, 65\}$ & 0.285714285714286      & \\
		73  & 20         & $\{10, 40, 1, 4, 9, 17, 21, 22, 24, 26, 28, 29, 32, 46, 53, 54, 57, 59, 61, 71\}$ & 0.277777777777778      &
	\end{tabular}
	\centering
	\caption{$M(0,3;q)$ for Small $q$}
	\label{tab-m03q}
\end{table*}
\begin{figure}
	\begin{tikzpicture}
		\SetVertexSimple[MinSize = .25cm]
		\Vertex[x=2.276cm,y=3.8752cm]{v0}
		\Vertex[x=3.5051cm,y=3.2318cm]{v1}
		\Vertex[x=2.4665cm,y=1.7403cm]{v2}
		\Vertex[x=3.0091cm,y=1.2934cm]{v3}
		\Vertex[x=3.7423cm,y=0.0cm]{v4}
		\Vertex[x=1.3156cm,y=0.956cm]{v5}
		\Vertex[x=0.632cm,y=4.3015cm]{v6}
		\Vertex[x=3.7057cm,y=0.5442cm]{v7}
		\Vertex[x=2.1395cm,y=2.7747cm]{v8}
		\Vertex[x=0.4139cm,y=0.8486cm]{v9}
		\Vertex[x=5.0cm,y=1.5695cm]{v10}
		\Vertex[x=1.6783cm,y=3.0754cm]{v11}
		\Vertex[x=2.3008cm,y=0.8298cm]{v12}
		\Vertex[x=2.7201cm,y=5.0cm]{v13}
		\Vertex[x=2.8351cm,y=3.2336cm]{v14}
		\Vertex[x=0.8424cm,y=2.2628cm]{v15}
		\Vertex[x=1.9325cm,y=0.2369cm]{v16}
		\Vertex[x=1.186cm,y=4.5212cm]{v17}
		\Vertex[x=3.8637cm,y=2.2432cm]{v18}
		\Vertex[x=3.7734cm,y=1.4238cm]{v19}
		\Vertex[x=1.7294cm,y=4.9142cm]{v20}
		\Vertex[x=0.0cm,y=2.6606cm]{v21}
		\Vertex[x=3.3497cm,y=4.7289cm]{v22}
		\Vertex[x=1.708cm,y=1.4332cm]{v23}
		\Vertex[x=3.0144cm,y=4.3086cm]{v24}
		\Vertex[x=4.2817cm,y=4.8313cm]{v25}
		\Vertex[x=4.9792cm,y=3.7538cm]{v26}
		\Vertex[x=4.4503cm,y=3.4016cm]{v27}
		\Edge(v2)(v0)
		\Edge(v6)(v0)
		\Edge(v0)(v1)
		\Edge(v0)(v25)
		\Edge(v0)(v13)
		\Edge(v0)(v15)
		\Edge(v1)(v22)
		\Edge(v1)(v25)
		\Edge(v11)(v1)
		\Edge(v2)(v16)
		\Edge(v2)(v3)
		\Edge(v2)(v4)
		\Edge(v2)(v1)
		\Edge(v2)(v15)
		\Edge(v3)(v4)
		\Edge(v3)(v5)
		\Edge(v3)(v7)
		\Edge(v3)(v1)
		\Edge(v3)(v11)
		\Edge(v4)(v16)
		\Edge(v4)(v19)
		\Edge(v4)(v7)
		\Edge(v4)(v10)
		\Edge(v5)(v21)
		\Edge(v5)(v7)
		\Edge(v5)(v9)
		\Edge(v5)(v11)
		\Edge(v5)(v12)
		\Edge(v6)(v17)
		\Edge(v6)(v20)
		\Edge(v6)(v21)
		\Edge(v6)(v13)
		\Edge(v6)(v15)
		\Edge(v7)(v18)
		\Edge(v7)(v10)
		\Edge(v7)(v12)
		\Edge(v8)(v18)
		\Edge(v8)(v20)
		\Edge(v8)(v23)
		\Edge(v8)(v24)
		\Edge(v8)(v12)
		\Edge(v8)(v14)
		\Edge(v9)(v16)
		\Edge(v9)(v21)
		\Edge(v9)(v23)
		\Edge(v9)(v12)
		\Edge(v9)(v15)
		\Edge(v10)(v18)
		\Edge(v10)(v19)
		\Edge(v10)(v26)
		\Edge(v10)(v27)
		\Edge(v11)(v17)
		\Edge(v11)(v21)
		\Edge(v11)(v22)
		\Edge(v12)(v18)
		\Edge(v12)(v23)
		\Edge(v13)(v20)
		\Edge(v13)(v25)
		\Edge(v13)(v27)
		\Edge(v13)(v14)
		\Edge(v14)(v19)
		\Edge(v14)(v20)
		\Edge(v14)(v23)
		\Edge(v14)(v27)
		\Edge(v15)(v16)
		\Edge(v15)(v21)
		\Edge(v16)(v19)
		\Edge(v16)(v23)
		\Edge(v17)(v20)
		\Edge(v17)(v21)
		\Edge(v17)(v22)
		\Edge(v17)(v24)
		\Edge(v18)(v24)
		\Edge(v18)(v26)
		\Edge(v19)(v23)
		\Edge(v19)(v27)
		\Edge(v20)(v24)
		\Edge(v22)(v24)
		\Edge(v22)(v25)
		\Edge(v22)(v26)
		\Edge(v24)(v26)
		\Edge(v25)(v26)
		\Edge(v25)(v27)
		\Edge(v26)(v27)
	\end{tikzpicture}\centering
	\caption{Cayley graph $\Cay(\mathbb{Z}_{29}^\times; [1,3]\cdot [1,3]^{-1})$}
	\label{fig-z29}
\end{figure}
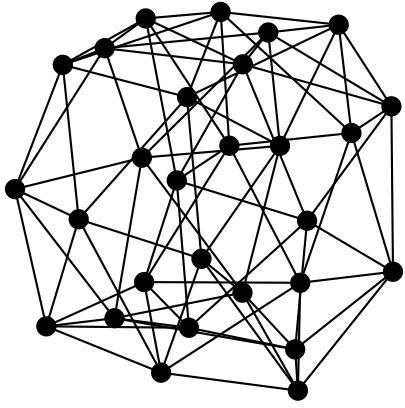
\begin{figure}\centering
	\begin{tikzpicture}
		\Vertex[LabelOut=false,L=\hbox{$16$},x=3.6367cm,y=0.0cm]{v0}
		\Vertex[LabelOut=false,L=\hbox{$10$},x=3.8732cm,y=4.6996cm]{v1}
		\Vertex[LabelOut=false,L=\hbox{$20$},x=1.338cm,y=5.0cm]{v2}
		\Vertex[LabelOut=false,L=\hbox{$15$},x=5.0cm,y=2.2325cm]{v3}
		\Vertex[LabelOut=false,L=\hbox{$1$},x=2.5072cm,y=2.4959cm]{v4}
		\Vertex[LabelOut=false,L=\hbox{$2$},x=0.0cm,y=2.7388cm]{v5}
		\Vertex[LabelOut=false,L=\hbox{$3$},x=1.1745cm,y=0.2867cm]{v6}
		\Edge(v6)(v0)
		\Edge(v3)(v0)
		\Edge(v4)(v0)
		\Edge(v1)(v2)
		\Edge(v1)(v3)
		\Edge(v4)(v1)
		\Edge(v4)(v2)
		\Edge(v4)(v3)
		\Edge(v5)(v6)
		\Edge(v5)(v2)
		\Edge(v4)(v5)
		\Edge(v4)(v6)
	\end{tikzpicture}
	\caption{Neighborhood of $1$ in $\Cay(\mathbb{Z}_{29}^\times; [1,3]\cdot [1,3]^{-1})$}
	\label{fig-z29-nbd}
\end{figure}
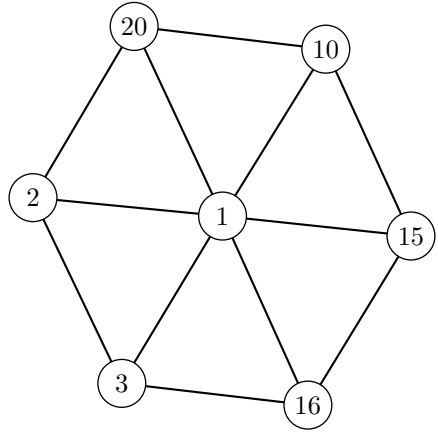
\begin{prob}
	Decide 
    $\liminf_{\substack{q\to\infty \\ q\text{ prime}}} M(0,3;q)/(q-1)$. Will it be greater than $1/4$?
\end{prob}
We have also tried to get an upper bound for $M(0,3;q)$, but failed to get an ideal bound. So we propose
\begin{prob}
	Develop suitable tools to get an effective upper bound for $M(0,3;q)$.
\end{prob}
\section{On quasi-perfect $B[-4,4](2p)$ sets for odd prime $p$}
In this section we discuss quasi-perfect $B[-4,4](2p)$ sets. A known construction is
\begin{constr}[{\cite[Theorem 4]{klove2012codes}}]
	\label{constr-4-4-2p}
	Let $p\equiv 1\pmod 4$ be a prime and fix an odd primitive root $g$ modulo $2p$.
	If both $\ind_g(2),\ind_g(3)$ are odd, then the set
	$$
		\{ g^{2i} \pmod{2p} \colon 0\le i <(p-1)/4\}
	$$
	is a quasi-perfect $B[-4,4](2p)$ set.
\end{constr}
We note that both $\ind_g(2)$ and $\ind_g(3)$ being odd is equivalent to $2$ and $3$ being quadratic nonresidues modulo $p$. Since $2$ is a quadratic nonresidue modulo $p$ if and only if $p \equiv 3,5 \pmod{8}$,
$3$ is a quadratic nonresidue modulo $p$ if and only if $p\equiv 5,7\pmod {12}$, the condition is equivalent to $p \equiv 5 \pmod{24}$.

We will show that there are many other primes that admit a quasi-perfect $B[-4,4](2p)$ set. Suppose $p$ is an odd prime such that there exists a quasi-perfect $B[-4,4](2p)$ set $B$. If $p=4k+1$, then
$$
	|(B\sqcup 2B\sqcup 3B\sqcup 4B)\sqcup(-B\sqcup -2B\sqcup -3B\sqcup -4B)|=2p-2,
$$
and $|B|=k$. There are $4k$ even numbers in $\mathbb{Z}_{2p}\setminus\{0\}$. But $2B\sqcup 4B\sqcup -2B\sqcup-4B$ contains $4k$ even numbers. It suggests that $B$ contains only odd numbers. Moreover, $p\notin B$ since $2p=0\notin 2B$. Thus
$$(B\sqcup 2B\sqcup 3B\sqcup 4B)\sqcup(-B\sqcup -2B\sqcup -3B\sqcup -4B)=\mathbb{Z}_{2p}\setminus\{0,p\},$$
where $B\sqcup -B\sqcup 3B\sqcup -3B=\mathbb{Z}_{2p}^\times$ and $2B\sqcup 4B\sqcup -2B\sqcup -4B=2\mathbb{Z}_{2p}\setminus\{0\}$.
Choose an odd primitive root $g$ modulo $2p$, which is also a primitive root modulo $p$, and let
$A:=\{\ind_g(i):i\in B\}$. Then
$$
	A+\{0,\ind_g(3)\}+\left\{0,\frac{p-1}2\right\}=\mathbb{Z}_{p-1}
$$
is a factorization. According to \cite[Lemma III.13]{yuan2025existence},
this holds if and only if $\overline{A}+\{0,\ind_g(3)\pmod{\frac{p-1}2}\}=\mathbb{Z}_{(p-1)/2}$ is a factorization,
where $\overline{A}=\{a\pmod{\frac{p-1}2}:a\in A\}$. We know from \cite[Corollary III.8]{yuan2025existence} or \cite[Corollary 2.7]{buratti1996packing} that such an $\overline{A}$ exists if and only if $v_2(\ind_g(3)\bmod{\frac{p-1}2})<v_2(\frac{p-1}2)$. %

Next, we prove that
$$B\sqcup 2B\sqcup -B\sqcup -2B\equiv \mathbb{Z}_{p}\setminus\{0\}\pmod{p}.$$
It suffices to prove that the four operands on the left-hand side are pairwise disjoint.
Suppose $i\in B\pmod p$ and $i\in 2B\pmod p$, then $2i\in 2B, 2i\in 4B$. Since $2B\cap 4B=\varnothing$, such $i$ does not exist, so $B$ and $2B$ are disjoint modulo $p$. Results on other pairs of operands can be proven similarly. We see that $B\pmod{p}$ is a perfect $B[-2,2](p)$ splitter set, which exists if and only if $v_2(\ord_p(2))\ge 2$ (See \cite[Corollary 3]{klove2012codes}), or equivalently $v_2(\ind_g(2)\bmod{\frac{p-1}2})<v_2(\frac{p-1}2)$.

Elements in $B$ are odd, and hence $\{\ind_g(i):i\in B\pmod{p}\}=\{\ind_g(i\bmod{p}):i\in B\}=\{\ind_g(i):i\in B\}=A$. Since $B\pmod{p}$ is a perfect $B[-2,2](p)$ splitter set,
$$
	A+\{0,\ind_g(2)\}+\left\{0,\frac{p-1}2\right\}=\mathbb{Z}_{p-1}
$$
is a factorization, which holds if and only if $\overline{A}+\{0,\ind_g(2)\pmod{\frac{p-1}2}\}=\mathbb{Z}_{(p-1)/2}$ is a factorization. Let $f(x)$ be the mask polynomial of $\overline{A}$, then
$$
	\begin{aligned}
		 & f(x)(1+x^{\ind_g(2)\bmod{\frac{p-1}2}})\equiv f(x)(1+x^{\ind_g(3)\bmod{\frac{p-1}2}}) \\\equiv&
		1+x+\cdots+x^{\frac{p-1}2-1}\pmod{x^{\frac{p-1}2}-1}.
	\end{aligned}
$$
By \cite[Lemma III.2]{yuan2025existence}, we can prove that
if $\mathfrak{M}_1=\{1\le j\le v_2(\frac{p-1}2):\Phi_{2^j}(x)\mid 1+x^{\ind_g(2)\bmod{\frac{p-1}2}}\}$
and $\mathfrak{M}_2=\{1\le j\le v_2(\frac{p-1}2):\Phi_{2^j}(x)\mid 1+x^{\ind_g(3)\bmod{\frac{p-1}2}}\}$,
then $\mathfrak{M}_1=\mathfrak{M}_2$ and $|\mathfrak{M}_1|=|\mathfrak{M}_2|=1$.
From this and \cite[Lemma III.1]{yuan2025existence}, we deduce that $v_2(\ind_g(2)\bmod{\frac{p-1}2})=v_2(\ind_g(3)\bmod{\frac{p-1}2})$. Since $v_2(\ind_g(2)\bmod{\frac{p-1}2})<v_2(\frac{p-1}2)$,
then $v_2(\ind_g(2)\bmod{\frac{p-1}2})=v_2(\ind_g(2))$. Similarly $v_2(\ind_g(3)\bmod{\frac{p-1}2})=v_2(\ind_g(3))$. So we can write the necessary condition for the existence of a quasi-perfect $B[-4,4](2p)$ set as: $v_2(\ind_g(2))=v_2(\ind_g(3))<v_2(\frac{p-1}2)$.

When $i_1:=v_2(\ind_g(2))=v_2(\ind_g(3))<v_2(\frac{p-1}2)$, we set $A=\{0,1,\ldots,2^{i_1}-1\}+\{0,2^{i_1+1},2\cdot 2^{i_1+1},\ldots,\frac{p-1}2\}$, then $|A|=2^{i_1}\cdot(\frac{p-1}2\big/2^{i_1+1})=\frac{p-1}4$, and $A$ is a common complementary factor of $\{0,\ind_g(2)\bmod{\frac{p-1}2}\}$ and $\{0,\ind_g(3)\bmod{\frac{p-1}2}\}$ in $\mathbb{Z}_{(p-1)/2}$. Let $B=\{g^i:i\in A\}$, then $B$ is a quasi-perfect $B[-4,4](2p)$ splitter set.

\begin{exampl}
  Pick $p=97, g=5,\ind_g(2)=34,\ind_g(3)=70$. Then $v_2(\ind_g(2))=v_2(\ind_g(3))<v_2(\frac{p-1}2)$. We can verify that
  $$
  \{g^{i+4j}:i\in[0,1],j\in[0,11]\}
  $$
  is a quasi-perfect $B[-4,4](194)$ splitter set.
\end{exampl}

Now assume that $p=4k+3\ge 7$ and that there is a quasi-perfect $B[-4,4](2p)$ set $B$,
then $|B|=k$ and as before $p\notin (B\sqcup 2B\sqcup 3B\sqcup 4B)\sqcup(-B\sqcup -2B\sqcup -3B\sqcup -4B)$.
The number of even numbers in $(B\sqcup 2B\sqcup 3B\sqcup 4B)\sqcup(-B\sqcup -2B\sqcup -3B\sqcup -4B)$
is $4k\le 4|B|+4|\{\text{even numbers in }B\}|\le p-1=4k+2$,
then we must have $|\{\text{even numbers in }B\}|=0$ and there are two even members $2a,2b$ of $\mathbb{Z}_{2p}$ such that $2a,2b\notin(B\sqcup 2B\sqcup 3B\sqcup 4B)\sqcup(-B\sqcup -2B\sqcup -3B\sqcup -4B)$.
Moreover, there are two odd members $2c+1,2d+1$ of $\mathbb{Z}_{2p}$ such that $2c+1,2d+1\notin(B\sqcup 2B\sqcup 3B\sqcup 4B)\sqcup(-B\sqcup -2B\sqcup -3B\sqcup -4B)$ and hence
$$
	\begin{aligned}
		 & \mathbb{Z}_{2p}\setminus\{0,p,2a,2b,2c+1,2d+1\} \\=&(B\sqcup 2B\sqcup 3B\sqcup 4B)\sqcup(-B\sqcup -2B\sqcup -3B\sqcup -4B).
	\end{aligned}
$$
Look at the even elements, we have
$$
	2\mathbb{Z}_{2p}\setminus\{0,2a,2b\}=2B\sqcup 4B\sqcup -2B\sqcup -4B.
$$
Since $2\mathbb{Z}_{2p}$ is a group and the right-hand side is closed under taking additive inverses, hence so is $\{0,2a,2b\}$, and $-2a\in\{0,2a,2b\}$.
However, $2a\ne p$, so $-2a=2b$. As before, we have
$$
	B\sqcup 2B\sqcup -B\sqcup -2B\sqcup\{\pm a\}\equiv \mathbb{Z}_p^\times\pmod{p}.
$$
Taking discrete logarithm in $\mathbb{Z}_p^\times$, we get
$$
	((A+\{0,\ind_g(2)\})\sqcup\{\ind_g(a)\})+\left\{0,\frac{p-1}2\right\}=\mathbb{Z}_{p-1}
$$
is a factorization, which holds if and only if $\overline{A}+\{0,\ind_g(2)\pmod{\frac{p-1}2}\}$ is a direct sum and equals $\mathbb{Z}_{(p-1)/2}\setminus\{\ind_g(a)\pmod{\frac{p-1}2}\}$,
where $\overline{A}=\{a\pmod{\frac{p-1}2}:a\in A\}$.

Now look at the odd elements. We have
$$
	B\sqcup 3B\sqcup -B\sqcup -3B\sqcup\{2c+1,2d+1\}=\mathbb{Z}_{2p}^\times.
$$
As before, $2d+1=-(2c+1)$, and we have
$$
	((A+\{0,\ind_g(3)\})\sqcup\{\ind_g(2c+1)\})+\left\{0,\frac{p-1}2\right\}=\mathbb{Z}_{p-1}
$$
is a factorization, which holds if and only if $\overline{A}+\{0,\ind_g(3)\pmod{\frac{p-1}2}\}$ is a direct sum and equals $\mathbb{Z}_{(p-1)/2}\setminus\{\ind_g(2c+1)\pmod{\frac{p-1}2}\}$.

From Proposition \ref{prop-2-uniqueness}, we must have $\ind_g(2)\equiv \pm\ind_g(3)\pmod{\frac{p-1}2}$.
Since $\ind_g(2)\ne\ind_g(3)$, so if $\ind_g(2)\equiv\ind_g(3)\pmod{\frac{p-1}2}$, then $\ind_g(2)\equiv\ind_g(3)+\frac{p-1}2\pmod{(p-1)}$, which implies $2\equiv -3\pmod p$ and $p=5$, a contradiction with $p=4k+3$. If $\ind_g(2)\equiv-\ind_g(3)\pmod{\frac{p-1}2}$, then $\ind_g(6)\equiv0,\frac{p-1}2\pmod{(p-1)}$. Since $p\ne 5$, we must have $6\equiv-1\pmod{p}$ and $p=7$. When $p=7$, $B=\{1\}$ is clearly a quasi-perfect $B[-4,4](2p)$ set.  Results of this section give Proposition \ref{prop-pm4}.

\section{Perfect $B[0,pr](q)$ splitters for primes $p,q,r$}
In this section, we settle the problem of determining the existence condition for perfect $B[0,6](q)$ sets, while $q\equiv1\pmod 6$ is a prime. We fix a primitive root $g$ modulo $q$, and investigate when there is a group factorization $\mathbb{Z}_{q-1}=\{\ind_g(i) : i\in[1,6]\}+C$.
We first find out what condition a set $A$ of size $pr$ should satisfy in order that a factorization $\mathbb{Z}_N=A+C$ exists, where $p,r$ are distinct primes.

We quote a result on integer tilings. We say a finite set $A\subset\mathbb{Z}$ tiles $\mathbb{Z}$ by translations if there is a set $T\subset\mathbb{Z}$ such that $\mathbb{Z}=A+T$ and the sums in $A+T$ are pairwise distinct.
In 1999, Coven and Meyerowitz gave the following condition of integer tilings:
\begin{thm}[Coven--Meyerowitz, \cite{coven1999tiling}]
	Let $S_A$ be the set of prime powers $p^\alpha$ such that $\Phi_{p^\alpha}(x)$ divides $f_A(x)$.
	There are two conditions:
	\begin{enumerate}[label=(T\arabic*)]
		\item $f_A(1)=\prod_{i\in S_A} \Phi_s(1)$,
		\item if $s_1,\ldots,s_k\in S_A$ are powers of different primes, then $\Phi_{s_1\cdots s_k}$ divides $f_A(x)$.
	\end{enumerate}
	Then,
	\begin{itemize}
		\item if $A$ satisfies (T1), (T2) then $A$ tiles $\mathbb{Z}$;
		\item if $A$ tiles $\mathbb{Z}$ then (T1) holds;
		\item if $A$ tiles $\mathbb{Z}$ and $|A|$ has at most two distinct prime factors then (T2) holds.
	\end{itemize}
\end{thm}

We also need a lemma of de Bruijn:
\begin{lemma}[{\cite[Theorem 2]{debruijn1953factorization}}]
	Let $p,r$ be different primes, and $n=p^\lambda r^\mu$,
	$\lambda\ge1$, $\mu\ge 1$.
	Assume that $f_A(x)$ is a polynomial with nonnegative integral coefficients,
	$\Phi_n(x)\mid f_A(x)$, and the degree of $f_A(x)$ is less than $n$.
	Then there are polynomials $P(x),R(x)$ with nonnegative integral coefficients
	such that
	$$
		f_A(x)=P(x)\frac{x^n-1}{x^{n/p}-1}+R(x)\frac{x^n-1}{x^{n/r}-1}.
	$$
\end{lemma}
Let $A\subset \mathbb{Z}_N$ be a set with $pr$ elements. First suppose that $A$ is a direct factor of $\mathbb{Z}_N$, then $A$ tiles $\mathbb{Z}$. As seen in \cite[Lemma III.2]{yuan2025existence}, since $|A|=pr$, there are $\lambda, \mu$ such that
\begin{enumerate}
  \item $\lambda\le v_p(N)$, $\mu\le v_r(N)$;
	\item Each of $\Phi_{p^\lambda}(x)$, $\Phi_{r^\mu}(x)$ divides $f_A(x)$;
	\item for $1\le i\le v_p(N)$, $1\le j\le v_r(N)$, $i\ne \lambda$, $j\ne \mu$, $\Phi_{p^i}$ or $\Phi_{r^j}$ does not divide $f_A(x)$.
\end{enumerate}
Moreover, by Coven--Meyerowitz condition, $\Phi_{p^\lambda r^\mu}(x)\mid f_A(x)$. Let $n=p^\lambda r^\mu\mid N$, $\overline{f_A}(x)$ be the remainder of $f_A(x)$ when divided by $x^n-1$. It is easy to verify that it is a polynomial with integral nonnegative coefficients, and each of $\Phi_{p^\lambda}(x)$, $\Phi_{r^\mu}(x)$, and $\Phi_{p^\lambda r^\mu}(x)$ divides $\overline{f_A}(x)$.
Using the lemma above, we obtain the equality in the lemma, with $f_A(x)$ replaced with $\overline{f_A}(x)$. Notice that
$\Phi_{p^\lambda}(x)\mid R(x)$, $\Phi_{r^\mu}(x)\mid
	P(x)$, the degree of $P(x)\le n/p=p^{\lambda-1}r^\mu$,
and the degree of $R(x)\le n/r=p^\lambda r^{\mu-1}$ since the leading coefficients of all the involved polynomials are positive.

Now we put $x=1$ into the equality, obtaining
$$
	pr=f_A(1)=\overline{f_A}(1)=pP(1)+rR(1),
$$
where the second equality holds since $1$ is a root of $x^n-1$. Since the coefficients of $P(x),R(x)$ are nonnegative integers,
and $\Phi_{p^\lambda}(1)=p\mid R(1)$, $\Phi_{r^\mu}(1)=r\mid P(1)$, we must have $P(1)=0$ or $R(1)=0$, which implies $P(x)=0$ or $R(x)=0$.

Without loss of generality, suppose that $R(x)=0$, then $\overline{f_A}(x)=P(x)\frac{x^n-1}{x^{n/p}-1}$, with $P(1)=r$. Now $\Phi_{r^\mu}(x)\mid P(x)$, then there is a polynomial with nonnegative integral coefficients $g(x)$ such that
$$
	P(x)\equiv g(x)\cdot\Phi_{r^\mu}(x)\pmod{x^{r^\mu}-1}.
$$
(See {\cite{debruijn1953factorization}} for a proof.) Consider the evaluation at $1$ again, we see that $g$ is a monomial. It means that: 
\begin{quotation}
there exist $\lambda\le v_p(N),\mu\le v_r(N)$ such that $A\pmod{p^\lambda r^\mu}$ (as a multiset) can be written as the sum of a set of size $r$, $A'$, and $A''=\{0,p^{\lambda-1}r^\mu, \ldots,p^{\lambda-1}(p-1)r^\mu\}$, and $A'\pmod{r^\mu}$ is an arithmetic progression with difference $r^{\mu-1}$. 
\end{quotation}
Notice that $A''$ is in fact a subgroup of $\mathbb{Z}_{p^\lambda r^\mu}$, thus $A\pmod{p^\lambda r^\mu}$ is the union of some cosets of $A''$ in $\mathbb{Z}_{p^\lambda r^\mu}$. Moreover, it is easy to verify that all elements of $A\pmod{p^\lambda r^\mu}$  are distinct [hence $A\pmod{p^\lambda r^\mu}$ is an ordinary $pr$-subset of $\mathbb{Z}_{p^\lambda r^\mu}$].

On the other hand, if $A$ satisfies the condition described above, then we can verify that
$$
	A+\{0,1,2,\ldots,r^{\mu-1}-1\}+\{0,r^\mu,\ldots,(p^{\lambda-1}-1)r^\mu\}=\mathbb{Z}_{p^\lambda r^\mu},
$$
is a factorization, that is, $A$ is a direct factor of $\mathbb{Z}_{p^\lambda r^\mu}$ and hence of $\mathbb{Z}_N$.

Now let $q\equiv1\pmod 6$ and $g$ be a fixed primitive root modulo $q$. 
By the discussion above, there exists a perfect $B[0,6](q)$ set if and only if there exist some positive integers $\mu$ and $\lambda$, such that ${2^\lambda3^\mu}\mid(q-1)$, and $A=\{\ind_g(i):i\in[1,6]\}\subset\mathbb{Z}_{q-1}$ modulo $2^\lambda3^\mu$ is
\begin{enumerate}
	\item $\{0,2^{\lambda-1}3^\mu\}+A'$, where $A'\equiv\{x,x+3^{\mu-1},x+2\cdot3^{\mu-1}\}\pmod{3^\mu}$ for some $x$; or
	\item $\{0,2^{\lambda}3^{\mu-1}, 2^{\lambda+1}3^{\mu-1}\}+A'$, where $A'\equiv\{x,x+2^{\lambda-1}\}\pmod{2^{\lambda}}$ for some $x$.
\end{enumerate}

Let $\alpha:=\ind_g(2),\beta:=\ind_g(3),\gamma:=\ind_g(5)$, then $A=\{0,\alpha,\beta,2\alpha,\gamma,(\alpha+\beta)\}$. In the former case, $A\equiv\{x,x,x,x,x,x\}\pmod{3^{\mu-1}}$ as a multiset. Since $0\in A$, $x\equiv0\pmod{3^{\mu-1}}$.
Since $A'+3^{\mu-1}\equiv A'\pmod{3^\mu}$, without loss of generality we suppose that $x=0$.
Then
$$
	\begin{aligned}
		A\equiv & \{w3^\mu,y3^\mu+3^{\mu-1}, z3^\mu+2\cdot3^{\mu-1}\}                               \\
		        & +\{0,2^{\lambda-1}3^\mu\} \\
		=       & (\{w3^\mu\}+\{0,2^{\lambda-1}3^\mu\})                                             \\
		        & \cup (\{y3^\mu+3^{\mu-1}\}+\{0,2^{\lambda-1}3^\mu\})                              \\
		        & \cup (\{z3^\mu+2\cdot3^{\mu-1}\}+\{0,2^{\lambda-1}3^\mu\}) \pmod{2^{\lambda}3^\mu},
	\end{aligned}
$$
for some $w,y,z$.
Let $A=A_1\sqcup A_2\sqcup A_3$ be a partition of $A$, such that the elements in $A_1$ are divisible by $3^\mu$, the elements in $A_2$ are congruent to $3^{\mu-1}$ modulo $3^\mu$,
the elements in $A_3$ are congruent to $2\cdot 3^{\mu-1}$ modulo $3^\mu$. Then
$$
\begin{aligned}
A_1&\equiv\{w3^\mu\}+\{0,2^{\lambda-1}3^\mu\}\pmod{2^{\lambda}3^\mu},\\
A_2&\equiv\{y3^\mu+3^{\mu-1}\}+\{0,2^{\lambda-1}3^\mu\}\pmod{2^{\lambda}3^\mu},\\
A_3&\equiv\{z3^\mu+2\cdot3^{\mu-1}\}+\{0,2^{\lambda-1}3^\mu\}\pmod{2^{\lambda}3^\mu}.\\
\end{aligned}
$$

Notice that $\ind_g(1)=0\in A_1$. It is easy to verify that in fact $A_1=\{0,2^{\lambda-1}3^\mu\}$. So, $A_1$ is a subgroup of $\mathbb{Z}_{2^\lambda3^\mu}$, while $A_2$, $A_3$ are its cosets. If $v$ is an element of $A_i$, then $v+2^{\lambda-1}3^\mu$ is the other element of $A_i$.

Now consider $\alpha=\ind_g(2)$. If $A_1\ni\alpha\equiv2^{\lambda-1} 3^{\mu}\pmod{2^{\lambda}3^\mu}$, then $\ind_g(4)=2\ind_g(2)\equiv 0\pmod{2^{\lambda}3^\mu}$, contradicting the fact that elements of $A\pmod{2^\lambda3^\mu}$ are distinct.

So $\alpha\in A_2$ or $A_3$. If $A_2\ni\alpha\equiv y3^\mu+3^{\mu-1}+i2^{\lambda-1}3^\mu\pmod{2^{\lambda}3^\mu}$ for some $i$, then $\ind_g(4)\equiv 2\alpha\equiv 2y3^\mu+2\cdot 3^{\mu-1}+(2i)2^{\lambda-1}3^\mu\equiv 2y3^\mu+2\cdot 3^{\mu-1}\pmod{2^{\lambda}3^\mu}$ and $\ind_g(4)\in A_3$.
We have three possibilities:

\begin{itemize}
	\item $\beta=\ind_g(3)\equiv 2^{\lambda-1}3^\mu\pmod{2^{\lambda}3^\mu}$ and $\beta\in A_1$, then
	      $\ind_g(6)\equiv\alpha+\beta\equiv y3^\mu+3^{\mu-1}+(i+1)2^{\lambda-1}3^\mu\pmod{2^{\lambda}3^\mu}$ and $\alpha+\beta\in A_2$ automatically holds. This forces
	      $\gamma=\ind_g(5)\in A_3$ and hence $\gamma\equiv 2y3^\mu+2\cdot 3^{\mu-1}+2^{\lambda-1}3^\mu\equiv 2\alpha+\beta\pmod{2^{\lambda}3^\mu}$.

		To summarize, this case happens if and only if there exist  $\lambda\le v_2(q-1),\mu\le v_3(q-1)$  such that $\alpha\equiv 3^{\mu-1}\pmod{3^\mu}$, $\beta\equiv 2^{\lambda-1}3^\mu\pmod{2^\lambda3^\mu}$, $\gamma\equiv2\alpha+\beta\pmod{2^\lambda3^\mu}$.
		We prove that these conditions are equivalent to saying that
	      $v_2(\beta)<v_2(q-1),
		      v_3(\alpha)<v_3(q-1), \alpha/3^{v_3(\alpha)}\equiv1\pmod 3,v_3(\alpha)<v_3(\beta),
		      v_2(\beta)<v_2(\gamma-2\alpha-\beta),$ and $
		      v_3(\alpha)<v_3(\gamma-2\alpha-\beta)$.
		Suppose that there exist $\lambda,\mu$ such that $\alpha\equiv 3^{\mu-1}\pmod{3^\mu}$, $\beta\equiv 2^{\lambda-1}3^\mu\pmod{2^\lambda3^\mu}$, $\gamma\equiv2\alpha+\beta\pmod{2^\lambda3^\mu}$, then $v_2(\beta)=\lambda-1<\lambda\le v_2(q-1)$, $v_3(\alpha)=\mu-1<\mu\le v_3(q-1)$, 
           Since $\alpha\equiv 3^{\mu-1}\pmod{3^\mu}$, $\alpha/3^{v_3(\alpha)}\equiv1\pmod 3$. Since $3^\mu\mid\beta$, $v_3(\beta)\ge\mu>v_3(\alpha)$. Since $2^\lambda3^\mu\mid(\gamma-2\alpha-\beta)$, we have $v_2(\beta)<v_2(\gamma-2\alpha-\beta), v_3(\alpha)<v_3(\gamma-2\alpha-\beta)$.
           Conversely, if $v_2(\beta)<v_2(q-1),
		      v_3(\alpha)<v_3(q-1), \alpha/3^{v_3(\alpha)}\equiv1\pmod 3,v_3(\alpha)<v_3(\beta),
		      v_2(\beta)<v_2(\gamma-2\alpha-\beta),$ and $
		      v_3(\alpha)<v_3(\gamma-2\alpha-\beta)$, then let $\lambda=v_2(\beta)+1\le v_2(q-1), \mu=v_3(\alpha)+1\le v_3(q-1)$.
           Since $\alpha/3^{v_3(\alpha)}\equiv1\pmod 3$, then $\alpha\equiv3^{\mu-1}\pmod{3^\mu}$.
           Since $v_2(\beta)=\lambda-1, v_3(\beta)\ge\mu$, $\beta\pmod{2^\lambda3^\mu}$ is an element in $\mathbb{Z}_{2^\lambda3^\mu}$ that is divisible by $3^\mu$ and $2^{\lambda-1}$ but not divisible by $2^\lambda$. The only possible value is $2^{\lambda-1}3^\mu$.
           Finally, since $v_2(\gamma-2\alpha-\beta)>v_2(\beta)=\lambda-1$, $v_3(\gamma-2\alpha-\beta)>v_3(\alpha)=\mu-1$, we have $2^\lambda3^\mu\mid(\gamma-2\alpha-\beta)$ and consequently $\gamma\equiv2\alpha+\beta\pmod{2^\lambda3^\mu}$. Hence, we prove the equivalence.

	\item $A_2\ni\beta=\ind_g(3)\equiv \alpha+2^{\lambda-1}3^\mu\equiv y3^\mu+3^{\mu-1}+(i+1)2^{\lambda-1}3^\mu\pmod{2^{\lambda}3^\mu}$.
	      Then $\ind_g(6)\equiv 2y3^\mu+2\cdot 3^{\mu-1}+2^{\lambda-1}3^\mu\pmod{2^{\lambda}3^\mu}$ and $\ind_g(6)\in A_3$. So $A_1\ni\gamma=\ind_g(5)\equiv 2^{\lambda-1}3^\mu\equiv\beta-\alpha\pmod{2^{\lambda}3^\mu}$.

	      These conditions are equivalent to saying that $v_2(\gamma)<v_2(q-1)$, $v_3(\alpha)<v_3(q-1)$,
	      $v_3(\alpha)<v_3(\gamma)$,
	      $v_2(\gamma)<v_2(\alpha+\gamma-\beta)$, $v_3(\alpha)<v_3(\alpha+\gamma-\beta)$, $\alpha/3^{v_3(\alpha)}\equiv1\pmod 3$.

	\item $A_3\ni\beta=\ind_g(3)\equiv 2y3^\mu+2\cdot 3^{\mu-1}+i2^{\lambda-1}3^\mu\pmod{2^{\lambda}3^\mu}$. Then
	      $\ind_g(6)\equiv \alpha+\beta\equiv ((3y+1)+(i+1)2^{\lambda-1})3^\mu\pmod{2^{\lambda}3^\mu}$ must be in $A_1$, so $\alpha+\beta\equiv 2^{\lambda-1}3^\mu\pmod{2^\lambda3^\mu}$;
	      and $\gamma=\ind_g(5)\in A_2$, so $\gamma\equiv \alpha+2^{\lambda-1}3^\mu\equiv 2\alpha+\beta \pmod{2^{\lambda}3^\mu}$.
	      These conditions are equivalent to saying that $v_2(\alpha+\beta)<v_2(q-1)$, $v_3(\alpha)<v_3(q-1)$,
	      $v_3(\alpha)<v_3(\alpha+\beta)$, $v_3(\alpha)<v_3(\gamma-2\alpha-\beta)$, $v_2(\alpha+\beta)<v_2(\gamma-2\alpha-\beta)$, $\alpha/3^{v_3(\alpha)}\equiv1\pmod 3$.
\end{itemize}

If $A_3\ni\alpha\equiv z3^\mu+2\cdot3^{\mu-1}+i2^{\lambda-1}3^\mu\pmod{2^{\lambda}3^\mu}$, then
$\ind_g(4)\equiv2\alpha\equiv (2z+3)3^\mu+3^{\mu-1}\pmod{2^{\lambda}3^\mu}$ and $\ind_g(4)\in A_2$, we can similarly obtain the following three possibilities:
\begin{itemize}
	\item $A_1\ni\beta=\ind_g(3)\equiv 2^{\lambda-1}3^\mu\pmod{2^{\lambda}3^\mu}$,
	      then $\ind_g(6)\equiv z3^\mu+2\cdot3^{\mu-1}+(i+1)2^{\lambda-1}3^\mu\pmod{2^{\lambda}3^\mu}$ is in $A_3$,
	      $\gamma=\ind_g(5)\equiv (2z+3)3^\mu+3^{\mu-1}+2^{\lambda-1}3^\mu\pmod{2^{\lambda}3^\mu}\equiv 2\alpha+\beta$ is in $A_2$. It is equivalent to saying that
	      $v_2(\beta)<v_2(q-1)$,
	      $v_3(\alpha)<v_3(q-1)$,
	      $v_3(\alpha)<v_3(\beta)$,
	      $v_2(\beta)<v_2(\gamma-2\alpha-\beta)$,
	      $\alpha/3^{v_3(\alpha)}\equiv2\pmod 3$,
	      $v_3(\alpha)<v_3(\gamma-2\alpha-\beta)$.
	\item $A_3\ni\beta=\ind_g(3)\equiv z3^\mu+2\cdot3^{\mu-1}+(i+1)2^{\lambda-1}3^\mu\pmod{2^{\lambda}3^\mu}$,
	      $\ind_g(6)\equiv \alpha+\beta\equiv (2z+3)3^\mu+3^{\mu-1}+2^{\lambda-1}3^\mu\pmod{2^{\lambda}3^\mu}$ is in $A_2$,
	      $\gamma=\ind_g(5)\equiv 2^{\lambda-1}3^{\mu}\pmod{2^{\lambda}3^\mu}$ is in $A_1$.
	      It is equivalent to saying that
	      $v_2(\gamma)<v_2(q-1)$,
	      $v_3(\alpha)<v_3(q-1)$,
	      $v_3(\alpha)<v_3(\gamma)$,
	      $v_2(\gamma)<v_2(\gamma+\alpha-\beta)$,
	      $\alpha/3^{v_3(\alpha)}\equiv2\pmod 3$,
	      $v_3(\alpha)<v_3(\gamma+\alpha-\beta)$.
	\item $A_2\ni\beta=\ind_g(3)\equiv (2z+3)3^\mu+3^{\mu-1}+2^{\lambda-1}3^\mu\pmod{2^{\lambda}3^\mu}$,
	      $\ind_g(6)\equiv\alpha+\beta\equiv (3z+3)3^\mu+3\cdot 3^{\mu-1}+(i+1)2^{\lambda-1}3^\mu\pmod{2^{\lambda}3^\mu}$ is in $A_1$.
	      So we must have $\alpha+\beta\equiv 2^{\lambda-1}3^\mu\pmod{2^\lambda 3^\mu}$.
	      Then $A_3\ni\gamma=\ind_g(5)\equiv \alpha+2^{\lambda-1}3^\mu\equiv2\alpha+\beta\pmod{2^\lambda 3^\mu}$.
	      These conditions are equivalent to saying that
	      $v_2(\alpha+\beta)<v_2(q-1)$,
	      $v_3(\alpha)<v_3(q-1)$,
	      $\alpha/3^{v_3(\alpha)}\equiv 2\pmod 3$,
	      $v_3(\alpha)<v_3(\alpha+\beta)$,
	      $v_3(\alpha)<v_3(\gamma-2\alpha-\beta)$,
	      $v_2(\alpha+\beta)<v_2(\gamma-2\alpha-\beta)$.
\end{itemize}

In the latter case, $A\equiv\{x,x,x,x,x,x\}\pmod{2^{\lambda-1}}$ as a multiset. Since $0\in A$, $x\equiv0\pmod{2^{\lambda-1}}$.
Since $A'+2^{\lambda-1}\equiv A'\pmod{2^\lambda}$, without loss of generality we suppose that $x=0$.
Then
$$
	\begin{aligned}
		A\equiv & \{0,2^{\lambda}3^{\mu-1}, 2^{\lambda+1}3^{\mu-1}\}+\{y\cdot2^\lambda,z\cdot2^\lambda+2^{\lambda-1}\} \\
		=       & (\{0,2^{\lambda}3^{\mu-1}, 2^{\lambda+1}3^{\mu-1}\}+\{y\cdot2^\lambda\})                             \\
		        & \cup (\{0,2^{\lambda}3^{\mu-1}, 2^{\lambda+1}3^{\mu-1}\}+\{z\cdot2^\lambda+2^{\lambda-1}\}) \pmod{2^{\lambda}3^\mu}.
	\end{aligned}
$$
Let $A=A_1\sqcup A_2$, where the $2$-adic valuations of the elements in $A_2$ are $\lambda-1$, and those of the elements in $A_1$ are at least $\lambda$. %
Then
$$
\begin{aligned}
A_1		\equiv       & \{0,2^{\lambda}3^{\mu-1}, 2^{\lambda+1}3^{\mu-1}\}+\{y\cdot2^\lambda\}\pmod{2^{\lambda}3^\mu}, \\
A_2\equiv		        & \{0,2^{\lambda}3^{\mu-1}, 2^{\lambda+1}3^{\mu-1}\}+\{z\cdot2^\lambda+2^{\lambda-1}\}\pmod{2^{\lambda}3^\mu}.
\end{aligned}
$$
Notice that $\ind_g(1)=0\in A_1$. It is easy to verify that in fact $A_1\equiv\{0,2^{\lambda}3^{\mu-1}, 2^{\lambda+1}3^{\mu-1}\}\pmod{2^{\lambda}3^\mu}$, which is a subgroup of $\mathbb{Z}_{2^\lambda3^\mu}$ and $A_2$ is its coset.

Now consider $\alpha=\ind_g(2)$. If $A_1\ni\alpha\equiv2^\lambda 3^{\mu-1}\pmod{2^{\lambda}3^\mu}$, then $\ind_g(4)=2\ind_g(2)\equiv 2^{\lambda+1} 3^{\mu-1}\pmod{2^{\lambda}3^\mu}$ is also in $A_1$. We have $A_1=\{0,\alpha,2\alpha\}$, hence $\{\ind_g(3),\ind_g(5),\ind_g(6)\}= A_2\equiv\{\beta,\gamma,\beta+\alpha\}\equiv
	\{0,\alpha,2\alpha\}+\{z\cdot2^\lambda+2^{\lambda-1}\}\pmod{2^{\lambda}3^\mu}$, which holds if and only if $\gamma\equiv \beta+2\alpha\pmod{2^{\lambda}3^\mu}$.

If $A_1\ni\alpha\equiv 2^{\lambda+1}3^{\mu-1}\pmod{2^\lambda 3^\mu}$, then $\ind_g(4)=2\ind_g(2)\equiv 2^{\lambda+2} 3^{\mu-1}\equiv 2^{\lambda} 3^{\mu-1}\pmod{2^{\lambda}3^\mu}$. So again we have $A_1=\{0,\alpha,2\alpha\}$, and $\{\ind_g(3),\ind_g(5),\ind_g(6)\}=A_2\equiv
	\{0,\alpha,2\alpha\}+\{z\cdot2^\lambda+2^{\lambda-1}\}\pmod{2^{\lambda}3^\mu}$. We obtain the same condition for $\gamma$: $\gamma\equiv \beta+2\alpha\pmod{2^{\lambda}3^\mu}$.

In both cases, the conditions are equivalent to saying that $v_3(\alpha)<v_3(q-1)$, $v_2(\beta)<v_2(q-1)$,
$v_2(\beta)<v_2(\alpha)$,
$v_2(\beta)<v_2(\gamma-\beta-2\alpha)$, $v_3(\alpha)<v_3(\gamma-\beta-2\alpha)$.

If $\alpha\in A_2$, then $\alpha\equiv z2^\lambda+2^{\lambda-1}+i2^\lambda 3^{\mu-1}\pmod{2^\lambda3^\mu}$ for some $i\in[0,2]$.
Then $2\alpha\equiv z2^{\lambda+1}+2^{\lambda}+(2i\bmod 3)2^\lambda 3^{\mu-1} \pmod{2^\lambda3^\mu}$ which should be in $A_1$. %
Notice that if $v_2(\beta)=\lambda-1$, then $v_2(\alpha+\beta)\ge\lambda$; if $v_2(\beta)\ge\lambda$, $v_2(\alpha+\beta)=\lambda-1$.
So we have the following two possibilities:
$A_1=\{0,2\alpha,\beta\}, A_2=\{\alpha,\alpha+\beta,\gamma\}$; or $A_1=\{0,2\alpha,\alpha+\beta\}, A_2=\{\alpha,\beta,\gamma\}$.

In the first case, we require that
$\{\alpha+\beta-\alpha,\gamma-\alpha\}\equiv\{2\alpha,\beta\}\pmod{2^\lambda 3^\mu}$, so $\gamma\equiv 3\alpha\pmod{2^\lambda 3^\mu}$. Moreover, $2\alpha+\beta\equiv 0\pmod{2^\lambda 3^\mu}$. These conditions are equivalent to
$v_2(\alpha)<v_2(q-1)$, $v_2(\alpha)<v_2(\beta)$, $ v_3(\beta)<v_3(q-1)$,
$v_2(\alpha)<v_2(\gamma-3\alpha)$,
$v_3(\beta)<v_3(\gamma-3\alpha)$,
$v_3(\beta)<v_3(2\alpha+\beta)$.
In the second case, we require that
$2\alpha\equiv 2(\alpha+\beta)\pmod{2^\lambda 3^\mu}$,
so $2\beta\equiv0\pmod{2^\lambda3^\mu}$. Since $\beta\not\equiv0\pmod{2^\lambda3^\mu}$, we must have $\beta\equiv 2^{\lambda-1}3^\mu \pmod{2^\lambda3^\mu}$.
Since $A_2\equiv A_1+\beta\pmod{2^\lambda3^\mu}$ and $2\beta\equiv0\pmod{2^\lambda3^\mu}$, we have
$\gamma\equiv2\alpha+\beta\pmod{2^\lambda3^\mu}$.
These conditions are equivalent to saying that
$v_2(\beta)=v_2(\alpha)<v_2(q-1),
 v_3(\alpha)<v_3(q-1), 
v_3(\alpha)<v_3(\beta),
v_2(\beta)<v_2(\gamma-2\alpha-\beta),
v_3(\alpha)<v_3(\gamma-2\alpha-\beta)$.
Notice that in this case,
$A$ is always an arithmetic progression with difference $2^{\lambda-1}3^{\mu-1}$, and this case is completely covered by the first case.

As a result of this section, we have proven Theorem \ref{thm:06}. We provide examples for each case of the theorem:
\begin{exampl}
		      Let $q=515701$, $g=11$, $\alpha=109623=3 \cdot 36541$, $\beta=121950=2 \cdot 3^{2} \cdot 5^{2} \cdot 271$,
		      $\gamma=506580=2^{2} \cdot 3 \cdot 5 \cdot 8443$,
		      $\gamma-2\alpha-\beta=2^{3} \cdot 3^{2} \cdot 2297$,
		      $A=\{0, 109623, 121950, 219246, 506580, 231573\}=\{0, 3, 18, 6, 24, 21\}\pmod{36}$.
		      Then $q$ satisfies case 1 of Theorem \ref{thm:06}.
		      Let $C=\{0,1,2\}+\{0,9\}+\{0,36,72,\ldots,515664\}$, then $A+C=\mathbb{Z}_{q-1}$ is a factorization,
		      and $B=\{g^{i+9j+36k}:i\in[0,2],j\in[0,1],k\in[0,14324]\}$ is a perfect $B[0,6](515701)$ splitter set.
	      \end{exampl}
\begin{exampl}
		      Let $q=2075041$, $g=19$, $\alpha=841566=2\cdot3\cdot11\cdot41\cdot311$, $\beta=1495884=2^2\cdot3\cdot13\cdot43\cdot223$, $\gamma=1792242=2\cdot3^2\cdot17\cdot5857$,
		      $\alpha+\gamma-\beta=1137924=2^2 \cdot 3^2\cdot 73 \cdot 433$, $A=\{0,841566,1495884,1683132,1792242,262410\}=\{0,30,12,24,18,6\}\pmod{36}$.
		      Then $q$ satisfies case 2 of Theorem \ref{thm:06}.
		      Let $C=\{0,1,2,3,4,5\}+\{0,36,\ldots,2075004\}$, then $A+C=\mathbb{Z}_{q-1}$ is a factorization.
		      Hence $B=\{g^{i+36j}:i\in[0,5],j\in[0,57639]\}$ is a perfect $B[0,6](2075041)$ splitter set.
	      \end{exampl}\begin{exampl}
		      Let $q=428401$, $g=19$, $\alpha=395400=2^{3} \cdot 3 \cdot 5^{2} \cdot 659$, $\alpha/3^{v_3(\alpha)}\equiv1\pmod3$, $\beta=102354=2 \cdot 3 \cdot 7 \cdot 2437$, $\gamma=14970=2 \cdot 3 \cdot 5 \cdot 499$, $\gamma-2\alpha-\beta=- 2^{3} \cdot 3^{2} \cdot 12197$, $A=\{0, 395400, 102354, 362400, 14970, 69354\}\equiv\{0, 12, 6, 24, 30, 18\}\pmod{2^23^2}$.
		      Then $q$ satisfies case 3 of Theorem \ref{thm:06}.
		      Let $C=\{0,1,2,3,4,5\}+\{0,36,\ldots,428364\}$, then
		      $A+C=\mathbb{Z}_{q-1}$ is a factorization.
		      So $B=\{g^{i+36j}:i\in[0,5],j\in[0,11899]\}$ is a perfect $B[0,6](428401)$ splitter set.
	      \end{exampl}
	      \begin{exampl}
	Let $q=187921$, $g=7$, $\alpha=32352 = 2^{5} \cdot 3 \cdot 337$, $\beta=143764 = 2^{2} \cdot 127 \cdot 283$,
	$\gamma=156772 = 2^{2} \cdot 7 \cdot 11 \cdot 509$, $\gamma-2\alpha-\beta=- 2^{4} \cdot 3^{2} \cdot 359$,
	$$A=\{0, 32352, 143764, 64704, 156772, 176116\}\equiv\{0, 24, 52, 48, 28, 4\}\pmod{72}.$$
	Then $q$ satisfies the case 4 of Theorem \ref{thm:06}. Let $C=\{0,1,2,3\}+\{0,8,16\}+\{0,72,\ldots,187848\}$, then $A+C=\mathbb{Z}_{q-1}$. Hence
	$B=\{g^{i+8j+72\ell}:i\in[0,3],j\in[0,2],\ell\in[0,2609]\}$ is a perfect $B[0,6](187921)$ set.
\end{exampl}
	      \begin{exampl}
		      Let $q=394129$, $g=13$, $\alpha=357666=2 \cdot 3 \cdot 59611$, $\beta=329496=2^{3} \cdot 3 \cdot 13729$, $\gamma=170514=2 \cdot 3^{2} \cdot 9473$, 
$\gamma-3\alpha=- 2^{2} \cdot 3^{2} \cdot 11 \cdot 43 \cdot 53$,
$2\alpha+\beta=2^{2} \cdot 3^{2} \cdot 29023$
		      , $A=\{0, 357666, 329496, 321204, 170514, 293034\}\equiv\{0, 6, 24, 12, 18, 30\}\pmod{2^23^2}$. Then $q$ satisfies the case 5 of Theorem \ref{thm:06}.
		      Let $C=\{0,1,2,3,4,5\}+\{0,36,\ldots,394092\}$, then
		      $A+C=\mathbb{Z}_{q-1}$ is a factorization.
		      So $B=\{g^{i+36j}:i\in[0,5],j\in[0,10947]\}$ is a perfect $B[0,6](394129)$ splitter set.
	      \end{exampl}
\section{New constructions of perfect limited magnitude burst correcting codes}
In this section, we consider limited magnitude burst correcting codes, as in \cite{wei2022perfect}.
All bursts in this section are cyclic ones, and all subscripts in this section are considered modulo $n$.

\subsection{A general framework for constructing perfect limited magnitude burst correcting codes}

We require a generalized version of group splittings.
\begin{dfn}[Generalized splitting, {\cite{buzaglo2012tilings}}]
	Let $M\subset\mathbb{Z}^n$. We say that $M$ splits an abelian group $G$ with a splitting sequence $\mathbf{s}=(s_0,s_1,\ldots,s_{n-1})\in G^n$
	if the set $\{\mathbf{m}\cdot\mathbf{s}:\mathbf{m}\in M\}$ contains distinct elements of $G$, where for a sequence $\mathbf{m}=(m_0,m_1,\ldots,m_{n-1})$ we denote
	$$
		\mathbf{m}\cdot\mathbf{s}=\sum_{i=0}^{n-1} m_is_i.
	$$
\end{dfn}

We say that a finite set $\mathcal{B} \subset \mathbb{Z}^n$ \emph{packs} $\mathbb{Z}^n$ by a lattice $\Lambda \subset \mathbb{Z}^n$ if for any two distinct vectors $\mathbf{v}, \mathbf{v}' \in \Lambda$, the translations $\mathbf{v} + \mathcal{B}$ and $\mathbf{v}' + \mathcal{B}$ are disjoint, i.e.,
$$
	(\mathbf{v} + \mathcal{B}) \cap (\mathbf{v}' + \mathcal{B}) = \varnothing.
$$
Furthermore, if $\mathcal{B} \subset \mathbb{Z}^n$ packs $\mathbb{Z}^n$ by $\Lambda$ and, in addition, the union of all such translations covers the entire space:
$$
	\bigcup_{\mathbf{v} \in \Lambda} (\mathbf{v} + \mathcal{B}) = \mathbb{Z}^n,
$$
then we say that $\mathcal{B}$ \emph{tiles} $\mathbb{Z}^n$ with the lattice $\Lambda$. In this context, $(k_2,k_1)$-limited magnitude $t$-error correcting codes can be obtained by lattice packings of $\mathbb{Z}^n$ by error balls $\mathcal{B}(n,t,k_2,k_1)$. If we further require the code to be \emph{perfect}, then the packing must also form a \emph{tiling} of the space.

\begin{thm}[{\cite[Corollary 1]{buzaglo2012tilings}}]
\label{thm:tiling}
	A lattice tiling of $\mathbb{Z}^n$ with the shape $M\subset\mathbb{Z}^n$ exists if and only if there exists an abelian group $G$ of order $|M|$ such that $M$ splits $G$.
\end{thm}

So if there is  an abelian group $G$ and a sequence $(s_0,\ldots,s_{n-1})$ such that sums $\alpha_i s_i+\sum_{l\in I} \alpha_{i+l}s_{i+l}$, where $I\subset[1,b-1]$, $\alpha_j\in[-k_1,k_2]^*$, are distinct, and $G=\bigcup_{i=0}^{n-1}\{\alpha_i s_i+\sum_{l\in I} \alpha_{i+l}s_{i+l}:I\subset[1,b-1], \alpha_j\in[-k_1,k_2]^*\}\sqcup\{0\}$, then there is a lattice tiling of $\mathbb{Z}^n$ with $\mathcal{E}^\circ(n,b,k_2,k_1)$. Let $s=k_1+k_2, d=s+1$, we calculate the size of such error ball as $|\mathcal{E}^\circ(n,b,k_2,k_1)|=sd^{b-1}n+1$. 

Let $n=bt$, $q=sd^{b-1}bt+1$ be a prime power, and fix a generator of $\mathbb{F}_q^\times$, $g$. Let $r=bsd^{b-1}$, $h=g^{r}$. Denote by $C_i^r$ the order-$r$ cyclotomic classes, namely, the cosets of $\langle h\rangle$ in $\mathbb{F}_q^\times$:
	$$
		C_i^r := g^i\langle h\rangle,\quad i\in[0,r-1].
	$$
We attempt to build such a sequence in $\mathbb{F}_q^\times$:
\begin{equation}
\label{eq:developed-sequence}
	\begin{aligned}
	&s_0=x_0=1,s_1=x_1,\ldots,s_{b-1}=x_{b-1},\\
	&s_{b}=h,s_{b+1}=hx_1,\ldots,s_{2b-1}=hx_{b-1},\\
	&s_{2b}=h^2,\ldots,\\
	&s_{(t-1)b}=h^{t-1},s_{(t-1)b+1}=h^{t-1}x_1,\ldots,s_{tb-1}=h^{t-1}x_{b-1}
	\end{aligned}
\end{equation}
that satisfies the condition.

Let $A=[-k_1,k_2]$. For $0\le j\le b-1$, define $\Gamma_0(A)=\{0\}$ and, for $j\ge1$,
\[
\Gamma_j(A)=\Gamma_j^{(0)}(A)\cup
        \bigcup_{\rho=1}^j \Gamma_j^{(\rho)}(A),
\]
where
\[
\Gamma_j^{(0)}(A)=
\left\{\sum_{i=0}^{j-1}\alpha_i x_i:\alpha_i\in A\right\},
\]
and
\begin{align*}
\Gamma_j^{(\rho)}(A)=
\Bigg\{&\sum_{i=\rho}^{j-1}\alpha_i x_i
        +h\sum_{i=0}^{\rho-1}\beta_i x_i:\alpha_i,\beta_i\in A,
        \;\beta_{\rho-1}\ne0\Bigg\}.
\end{align*}
Set
\[
        R_0=A^*,
        R_j=\{\lambda x_j+\gamma:
        \lambda\in A^*,\ \gamma\in\Gamma_j(A)\}
        \quad (1\le j\le b-1).
\]
Here we see $x_i$ as indeterminates and hence the sums written above are formal expressions.
\begin{lemma}
\label{lem:gamma-count}
If the formal expressions in $\Gamma_j(A)$ are pairwise distinct, then
\[
        |\Gamma_j(A)|=d^j+jsd^{j-1}
        \qquad (j\ge1).
\]
Consequently,
\[
        \sum_{j=0}^{b-1}|R_j|
        =s\sum_{j=0}^{b-1}|\Gamma_j(A)|
        =bs d^{b-1}=r.
\]
\end{lemma}

\begin{lemma}
\label{lem:burst-sum}
With the notation above, the set of all nonzero cyclic $b$-burst sums generated by the sequence \eqref{eq:developed-sequence}, i.e. $\bigcup_{i=0}^{n-1}\{\alpha_i s_i+\sum_{l\in I} \alpha_{i+l}s_{i+l}:I\subset[1,b-1], \alpha_\rho\in[-k_1,k_2]^*\}$, is exactly
\[
        \left(\bigcup_{j=0}^{b-1}R_j\right)C_0^r.
\]
\end{lemma}

\begin{proof}
For $0\le i\le b-1$, write
$s_{mb+i}=h^m x_i$, where the block index \(m\) is taken modulo \(t\). Since \(h^m\in C_0^r\), it is enough to describe the factor inside one or two consecutive base
blocks.

Consider a nonzero cyclic \(b\)-burst sum. If its support is contained in
the \(m\)-th base block, let \(j\) be the largest local index whose
coefficient is nonzero. Then the sum has the form $h^m\left(\lambda x_j+\sum_{i=0}^{j-1}\alpha_i x_i\right)$,
where \(\lambda\in A^*\) and \(\alpha_i\in A\). Since
$\sum_{i=0}^{j-1}\alpha_i x_i\in\Gamma_j^{(0)}(A)$,
this sum lies in \(R_jC_0^r\).

Otherwise, the support meets two consecutive base blocks, say the \(m\)-th
and the \((m+1)\)-st blocks. Let \(\rho-1\) be the largest local index in
the \((m+1)\)-st block whose coefficient is nonzero, and let \(j\) be the
largest local index in the \(m\)-th block whose coefficient is nonzero.
Since the burst has length \(b\), we have \(1\le \rho\le j\le b-1\). Then
the sum has the form
\[
h^m\left(
\lambda x_j+
\sum_{i=\rho}^{j-1}\alpha_i x_i
+
h\sum_{i=0}^{\rho-1}\beta_i x_i
\right),
\]
where \(\lambda\in A^*\), \(\alpha_i,\beta_i\in A\), and
\(\beta_{\rho-1}\ne0\). The expression after \(\lambda x_j\) belongs to
\(\Gamma_j^{(\rho)}(A)\), so the burst sum again lies in \(R_jC_0^r\).
Thus every nonzero cyclic \(b\)-burst sum is contained in $\left(\bigcup_{j=0}^{b-1}R_j\right)C_0^r$.

Conversely, take an element of
$\left(\bigcup_{j=0}^{b-1}R_j\right)C_0^r$.
It has the form $h^m(\lambda x_j+\gamma)$,
where $\lambda\in A^*, \gamma\in\Gamma_j(A)$.
If \(j=0\), then \(\gamma=0\), and
$h^m\lambda=\lambda s_{mb}$ is a nonzero cyclic \(b\)-burst sum.

Now suppose \(j\ge1\). If \(\gamma\in\Gamma_j^{(0)}(A)\), then
$\gamma=\sum_{i=0}^{j-1}\alpha_i x_i$
for some \(\alpha_i\in A\). Hence
\[
h^m(\lambda x_j+\gamma)=\lambda s_{mb+j}+
\sum_{i=0}^{j-1}\alpha_i s_{mb+i},
\]
which is a cyclic \(b\)-burst sum.

If \(\gamma\in\Gamma_j^{(\rho)}(A)\) for some \(1\le \rho\le j\), then
\[
\gamma=
\sum_{i=\rho}^{j-1}\alpha_i x_i
+
h\sum_{i=0}^{\rho-1}\beta_i x_i,
\qquad
\beta_{\rho-1}\ne0.
\]
Therefore
\[
h^m(\lambda x_j+\gamma)
=
\lambda s_{mb+j}
+
\sum_{i=\rho}^{j-1}\alpha_i s_{mb+i}
+
\sum_{i=0}^{\rho-1}\beta_i s_{(m+1)b+i}.
\]
Its support is contained in the cyclic interval
$\{mb+\rho,mb+\rho+1,\ldots,(m+1)b+\rho-1\}$,
which has length \(b\). Hence it is also a cyclic \(b\)-burst sum. This completes the proof.
\end{proof}

For a fixed label $j$ and $\delta\in\Delta_j$, define $\Delta_j=\bigcup_{\lambda\in A^*}\lambda^{-1}\Gamma_j(A)$
and $\Lambda_j(\delta)=\{\lambda\in A^*:
        \lambda\delta\in\Gamma_j(A)\}$.

Elements of $\Delta_j$ will be called \emph{roots}, and elements of $\Lambda_j(\delta)$ will be called \emph{multipliers admitted by the root $\delta$}. Notice that non-leading coefficients of a root are of the form $\mu\lambda^{-1}$, where $\mu\in A, \lambda\in A^*$.
We have
\[
        R_j=\bigcup_{\delta\in\Delta_j}
        \{\lambda(x_j+\delta):\lambda\in\Lambda_j(\delta)\}.
\]
Assume that each nonzero coefficient $\lambda$ has a prescribed cyclotomic label $\sigma(\lambda)$, i.e.,
\begin{equation}
\label{eq:sigma-label}
        \lambda\in C_{\sigma(\lambda)}^r,
        \qquad \sigma:A^*\to\mathbb{Z}_r.
\end{equation}
For each $\Lambda\subseteq A^*$, define the associated cyclotomic block
\[
        B(\Lambda)=\{\sigma(\lambda):\lambda\in\Lambda\}\subseteq\mathbb{Z}_r.
\]

In order to choose the desired $x_i$, we utilize a corollary of Weil's bound on multiplicative character sums, which has been used in many combinatorial constructions \cite{buratti2009combinatorial,zhang2017splitter}.
\begin{thm}[{\cite[Corollary 31]{zhang2017splitter}}]
	\label{thm-weil}
	Let $q\equiv 1\pmod{r}$ be a prime power. %
	Let $\{b_1,\ldots,b_m\}$ be an arbitrary $m$-subset of $\mathbb{F}_q$ (and hence $b_i$ are pairwise distinct) and $(\beta_1,\ldots,\beta_m)$ be an arbitrary element of $\mathbb{Z}_r^m$. Then if $q$ is sufficiently large, namely $q>Q(r,m)$,
	we can always find $x\in\mathbb{F}_q$ such that $x+b_i\in C_{\beta_i}^r$, $i\in [1,m]$.
\end{thm}
In fact, we can avoid a bounded number of elements in the theorem above:
\begin{cor}
\label{lem:placement}
Let $r$ be a fixed positive integer.  Suppose $q\equiv1\pmod r$ is a prime power.  Let $b_1,\ldots,b_m\in\mathbb{F}_q$ be pairwise distinct and let $\beta_1,\ldots,\beta_m\in\mathbb{Z}_r$ be prescribed.  If $E\subset\mathbb{F}_q$, $|E|<L$ is bounded, then there exists $x\in\mathbb{F}_q\setminus E$ such that
\[
        x+b_i\in C_{\beta_i}^r,
        \qquad 1\le i\le m,
\]
when $q$ is sufficiently large.
\end{cor}
\begin{proof}
Choose $u\in\mathbb{F}_q^\times$ so that $u-e\notin\{b_1,\ldots,b_m\}$ for every $e\in E$, and choose a class $C_t^r$ not containing $u$. It is possible if $q>|E|m+1$.  %
If $q>Q(r,m+mL)$, we can find $x\in\mathbb{F}_q$ such that $x+b_i\in C_{\beta_i}^r$ for $1\le i\le m$ and $x+(u-e)\in C_t^r$ for $e\in E$.
 We cannot have $x=e\in E$, since then $x+(u-e)=u\in C_t^r$, a contradiction.
\end{proof}
We will use Corollary \ref{lem:placement} to decide $x_1,x_2,\ldots,x_{b-1}$ recursively. To be specific, we have the following theorem.
\begin{thm}[Cyclotomic packing criterion]
\label{thm:main-framework}
Keep the notation above, and fix the coefficient labels \eqref{eq:sigma-label}.  If
\begin{quotation}
(*) there are $u_{j,\delta}$ with $u_{0,0}=0$ such that translates
\[
        u_{j,\delta}+B(\Lambda_j(\delta))\subseteq\mathbb{Z}_r,
        \qquad 0\le j\le b-1,
        \quad \delta\in\Delta_j,
\]
are pairwise disjoint and the total cardinality of them is $r$, i.e. form a partition of $\mathbb{Z}_r$,
\end{quotation}
then, for all sufficiently large $q\equiv1\pmod r$, one can choose $x_1,\ldots,x_{b-1}\in\mathbb{F}_q$ so that $\bigcup_{j=0}^{b-1}R_j$ is a complete representative system for the $r$ cyclotomic classes $C_0^r,\ldots,C_{r-1}^r$.  Consequently, there is a perfect $(k_2,k_1)$-limited magnitude $b$-burst correcting code of length $n=\frac{q-1}{sd^{b-1}}$.
\end{thm}

\begin{proof}
Notice that \(x_0=1\). We prove that we can decide
\(x_1,x_2,\ldots,x_{b-1}\) recursively, such that
\[
x_j+\delta\in C^r_{u_{j,\delta}},
\qquad 0\le j\le b-1,\quad \delta\in\Delta_j.
\]
Since \(x_0=1\in C^r_0\), \(\Delta_0=\{0\}\), \(u_{0,0}=0\),
this condition is satisfied when \(j=0\). Moreover, the elements of
\(\Delta_1\) are scalar multiples of \(x_0\). Hence, after substituting
\(x_0=1\), distinct roots in \(\Delta_1\) take distinct values.

Suppose that \(x_1,x_2,\ldots,x_{j-1}\) are determined, such that
expressions in \(\Delta_j\) take pairwise distinct values when evaluating
at these values, and $x_k+\delta\in C^r_{u_{k,\delta}}$ for \(k\in[0,j-1]\) and \(\delta\in\Delta_k\). We want to find \(x_j\)
such that
\[
x_j+\delta\in C^r_{u_{j,\delta}}
\qquad
\text{for every }\delta\in\Delta_j.
\]
If \(j<b-1\), we also require that expressions in \(\Delta_{j+1}\) take
pairwise distinct values after \(x_j\) is chosen.

After substituting the already chosen values
\(x_0=1,x_1,\ldots,x_{j-1}\), every element
\(\eta\in\Delta_{j+1}\) becomes an affine linear expression in the single
variable \(X=x_j\): $\eta(X)=a_\eta X+c_\eta .$ Moreover, deleting the \(x_j\) term from a root in \(\Delta_{j+1}\)
leaves a root in \(\Delta_j\). In the wrap around case we use the inclusion $h\Gamma_j^{(0)}(A)\subseteq \Gamma_j(A)$.

Thus each constant term \(c_\eta\) is the value of some root in
\(\Delta_j\).

Define
\[
E_j=
\left\{
x\in\mathbb F_q
:
\eta(x)=\eta'(x)
\text{ for some distinct }\eta,\eta'\in\Delta_{j+1}
\right\}.
\]
If $\eta(X)=a_\eta X+c_\eta,
\eta'(X)=a_{\eta'}X+c_{\eta'}$,
and \(a_\eta\ne a_{\eta'}\), then the equation
\(\eta(X)=\eta'(X)\) has at most one solution. If
\(a_\eta=a_{\eta'}\), then equality would force
\(c_\eta=c_{\eta'}\). Since the constant terms are values of roots in
\(\Delta_j\), the induction hypothesis implies that this can happen only
when the corresponding roots in \(\Delta_j\) are the same. In that case
the two affine roots are the same, contrary to \(\eta\ne\eta'\). Hence no
exceptional value arises from such a pair. Therefore
\[
|E_j|\le \binom{|\Delta_{j+1}|}{2},
\]
which is bounded only in terms of \(b,k_1,k_2\), and is independent of
\(q\). If \(j=b-1\), put \(E_{b-1}=\varnothing\).

Since the values of the roots in \(\Delta_j\) are pairwise distinct by
the induction hypothesis, by Corollary \ref{lem:placement}, for sufficiently large
\(q\), we can choose
$x_j\in\mathbb F_q\setminus E_j$
such that
\[
x_j+\delta\in C^r_{u_{j,\delta}}
\qquad
\text{for every }\delta\in\Delta_j.
\]
If \(j<b-1\), the condition \(x_j\notin E_j\) guarantees that expressions
in \(\Delta_{j+1}\) take pairwise distinct values after \(x_j\) is chosen.
Thus we can find \(x_1,\ldots,x_{b-1}\) recursively, such that
\[
x_j+\delta\in C^r_{u_{j,\delta}},
\qquad
0\le j\le b-1,\quad \delta\in\Delta_j.
\]

Now we prove that $\bigcup_{j=0}^{b-1}R_j$
is a complete representative system for the \(r\) cyclotomic classes.

For fixed \(j\), the map
\[
A^*\times \Gamma_j(A)
\longrightarrow
\{(\delta,\lambda):\delta\in\Delta_j,\ \lambda\in\Lambda_j(\delta)\},
\qquad
(\lambda,\gamma)\longmapsto(\lambda^{-1}\gamma,\lambda)
\]
is a bijection. Hence
\[
\sum_{\delta\in\Delta_j}|\Lambda_j(\delta)|
=
|A^*|\,|\Gamma_j(A)|
=
s|\Gamma_j(A)|.
\]
By the construction above, the roots in \(\Delta_j\) take pairwise
distinct values. Since \(1\in A^*\), we have
\(\Gamma_j(A)\subseteq \Delta_j\). Hence the expressions in
\(\Gamma_j(A)\) are also pairwise distinct. By Lemma~\ref{lem:gamma-count},
\[
\sum_{j=0}^{b-1}\sum_{\delta\in\Delta_j}|\Lambda_j(\delta)|
=
s\sum_{j=0}^{b-1}|\Gamma_j(A)|
=
bs d^{b-1}
=
r.
\]

On the other hand, assumption \((*)\) says that
\[
u_{j,\delta}+B(\Lambda_j(\delta)),
\qquad
0\le j\le b-1,\quad \delta\in\Delta_j,
\]
form a partition of \(\mathbb Z_r\). Hence
$r=\sum_{j=0}^{b-1}\sum_{\delta\in\Delta_j}
|B(\Lambda_j(\delta))|$. Since $|B(\Lambda_j(\delta))|\le |\Lambda_j(\delta)|$, we have
$r=
\sum_{j,\delta}|B(\Lambda_j(\delta))|
\le
\sum_{j,\delta}|\Lambda_j(\delta)|
=r$.
Thus each equality must be attained. In particular, $|B(\Lambda_j(\delta))|=|\Lambda_j(\delta)|$
for every \(j\) and every \(\delta\in\Delta_j\). Equivalently, the map
\(\sigma\) is injective on every \(\Lambda_j(\delta)\).

Now for \(0\le j\le b-1\), \(\delta\in\Delta_j\), and
\(\lambda\in\Lambda_j(\delta)\), we have
\[
x_j+\delta\in C^r_{u_{j,\delta}},
\qquad
\lambda\in C^r_{\sigma(\lambda)}.
\]
Therefore
$\lambda(x_j+\delta)\in C^r_{u_{j,\delta}+\sigma(\lambda)}$.
Since
\[
u_{j,\delta}+B(\Lambda_j(\delta)),
\qquad
0\le j\le b-1,\quad \delta\in\Delta_j,
\]
form a partition of \(\mathbb Z_r\), and since \(\sigma\) is injective on
each \(\Lambda_j(\delta)\), we see that
$\bigcup_{j=0}^{b-1}R_j$
is a complete representative system for the \(r\) cyclotomic classes
$C^r_0,C^r_1,\ldots,C^r_{r-1}$.

By Lemma~\ref{lem:burst-sum}, the set of all nonzero cyclic \(b\)-burst sums generated
by the sequence is
$\left(\bigcup_{j=0}^{b-1}R_j\right)C^r_0$.
Since \(\bigcup_{j=0}^{b-1}R_j\) is a complete representative system for
the \(r\) cyclotomic classes, this product is exactly $\mathbb F_q^\times$, and each element of $\mathbb{F}_q^\times$ has a unique representation, so we obtain a generalized
splitting of the additive group of \(\mathbb F_q\). By Theorem~\ref{thm:tiling}, in that case, there is a lattice tiling of
\(\mathbb Z^n\) with \(\mathcal{E}^\circ(n,b,k_2,k_1)\). Since
$|C^r_0|=\frac{q-1}{r}
\quad\text{and }
r=bs d^{b-1}$,
the corresponding length is
$n=b|C^r_0|
=\frac{b(q-1)}{r}=
\frac{q-1}{s d^{b-1}}$.
Hence there is a perfect \((k_2,k_1)\)-limited magnitude \(b\)-burst
correcting code of length
$n=\frac{q-1}{s d^{b-1}}$.
\end{proof}
\subsection{Arithmetic labelings}
\label{sec:arith}

This subsection records the arithmetic input used in the packing arguments for the constructions later.  The goal is to make selected rational numbers non-$r$th powers modulo primes; moreover, in the prime-field constructions for the families $(4,0)$ and $(3,1)$, one also needs controlled congruences between their cyclotomic labels. 
We also require that the primes are of the form $rk+1$. 
It can be achieved using a Kummer--Mills type result, see \cite{lietzmann1905zur, mills1963characters}. 
For completeness, we list the necessary notations and prove the forms we need in this paper.

\subsubsection{Exponent vectors}

Let $a_1,\ldots,a_m\in\mathbb{Q}^\times$.  Let $\pi_1,\ldots,\pi_N$ be the rational primes appearing in the reduced numerators or denominators of these numbers.  Write
\[
        a_i=\varepsilon_i\prod_{\nu=1}^N\pi_\nu^{e_{\nu i}},
        \qquad \varepsilon_i\in\{\pm1\},\quad e_{\nu i}\in\mathbb{Z}.
\]
For an odd prime $\ell$, the sign is irrelevant modulo $\ell$th powers, since $-1=(-1)^\ell$.  The \emph{exponent vector} of $a_i$ modulo $\ell$ is
\[
        v(a_i)=(e_{1i},\ldots,e_{Ni})\pmod\ell\in\mathbb{F}_\ell^N.
\]
Multiplication in $\mathbb{Q}^\times$ corresponds to addition of exponent vectors:
\[
        v(ab)=v(a)+v(b),
        \qquad v(a/b)=v(a)-v(b).
\]

\subsubsection{Kummer--Chebotarev labelings}

The following statement is the main tool to control labelings, or equivalently $\ell$-th power residue characters, and is in the form used throughout the paper.  The case $\eta=0$ gives ordinary primes $p\equiv1\pmod r$.  The case $\eta=1$ is the signed version; it imposes the congruence $p\equiv1+r\pmod{2r}$ and is used when the coefficient set contains negative integers.

\begin{thm}
\label{thm:kummer-labeling}
Let $r\ge2$, and let $\ell\mid r$ be an odd prime.  Let $\eta\in\{0,1\}$; if $\eta=1$, assume that $r$ is even. Prescribe a vector $u=(u_1,u_2,\dots,u_N)\in\mathbb{Z}_\ell^N$. Define
\[
R_\eta=\begin{cases}
r, & \eta=0,\\
2r, & \eta=1,
\end{cases}
\qquad
A_\eta=\begin{cases}
1, & \eta=0,\\
1+r, & \eta=1.
\end{cases}
\]
Let $u\in\mathbb{F}_\ell^N$.  Then there are infinitely many rational primes $p$ such that
\[
        p\equiv A_\eta\pmod{R_\eta},
\]
avoiding the primes appearing in the numerators and denominators of the $a_i$, and, if we fix a primitive root $g$ modulo $p$ and
\[
        a_i\in C_{\alpha_i}^r\subset\mathbb{F}_p^\times,
\]
then for some scalar $c_p\in\mathbb{F}_\ell^\times$ independent of $i$,
\[
        \alpha_i\equiv c_p\,u\cdot v(a_i)\pmod\ell
        \qquad (1\le i\le m).
\]
\end{thm}

\begin{proof}
Let $K=\mathbb{Q}(\zeta_r)$, where $\zeta_\ell$ is a fixed $\ell$th root of unity. Since $\ell\mid r$, we have $\zeta_\ell\in K$.  Choose $\alpha_\nu=\sqrt[\ell]{\pi_\nu}$ and set
\[
        L=K(\alpha_1,\ldots,\alpha_N).
\]
The rational primes $\pi_1,\ldots,\pi_N$ remain linearly independent in $K^\times/(K^\times)^\ell$.  Indeed, if a rational number $a\in\mathbb{Q}^\times$ becomes an $\ell$th power in the abelian field $K$, then $a$ was already an $\ell$th power in $\mathbb{Q}$: otherwise $X^\ell-a$ is irreducible over $\mathbb{Q}$, and the extension field $\mathbb{Q}(\sqrt[\ell]{a})/\mathbb{Q}$ has degree $\ell$, and if this field were contained in the abelian Galois field $K$ it would also be Galois; this would force $\zeta_\ell\in\mathbb{Q}(\sqrt[\ell]{a})$, impossible since $[\mathbb{Q}(\zeta_\ell):\mathbb{Q}]=\ell-1$ and $\ell-1\nmid\ell$.

Thus Kummer theory \cite[Chapter IV, \S 3, Theorem 3.3]{neukirch1999algebraic} tells us $\Gal(L/K)\cong(\mathbb{Z}/\ell\mathbb{Z})^N$ and gives an automorphism $\tau_u\in\Gal(L/K)$ satisfying
\[
        \tau_u(\alpha_\nu)=\zeta_\ell^{u_\nu}\alpha_\nu,
        \qquad 1\le \nu\le N.
\]
For $\eta=0$, put $K_\eta=K$, $M=L$.  For $\eta=1$, put $K_\eta=\mathbb{Q}(\zeta_{2r})$, $M=LK_\eta=L(\zeta_{2r})$ and since $[L:K]=\ell^N$, $[K_\eta:K]=2$, we have $[M:L]=2$ and may extend $\tau_{u}$ to $L_\eta$ such that 
\[
        \tau_u(\zeta_{2r})=\zeta_{2r}^{1+r},
\]
which fixes $\zeta_r$, and hence fixes $\zeta_\ell$ for every odd $\ell\mid r$.  %

We see $\tau_u$ as an element in $\Gal(M/\mathbb{Q})$. Now Chebotarev's density theorem \cite[Chapter VII, \S 13, Theorem 13.4]{neukirch1999algebraic}, applied to $M/\mathbb{Q}$ and $\tau_u$, guarantees that there are infinitely rational primes $p$ with an associated prime ideal $\mathfrak{P}\subset\mathcal{O}_M$ above $(p)$, such that
$$
x^{p}\equiv\tau_u(x)\pmod{\mathfrak{P}}
$$
for any $x\in\mathcal{O}_M$. Moreover, we can avoid $\pi_\nu$s, i.e. $p\notin\{\pi_1,\pi_2,\ldots,\pi_N\}$.

We have
$$
\alpha_\nu^p\equiv\tau_u(\alpha_\nu)=\zeta_\ell^{u_\nu}\alpha_\nu\pmod{\mathfrak{P}}
$$
for $1\le \nu\le N$. Since $p\ne\pi_\nu$ for any $\nu$, $\alpha_\nu$ is nonzero modulo $\mathfrak{P}$ and we may divide by $\alpha_\nu$ to obtain
$$
\pi_\nu^{(p-1)/\ell}=\alpha_\nu^{p-1}\equiv\zeta_{\ell}^{u_\nu}\pmod{\mathfrak{P}}.
$$

Since $\zeta_{R_\eta}\in \mathcal{O}_{K_\eta}$, and $\zeta_{R_\eta}^k-1\notin\mathfrak{p}, 1\le k\le {R_\eta}-1$ holds for all but a finite number of prime ideals $\mathfrak{p}\subset O_{K_{\eta}}$, we may avoid those rational primes lying below these $\mathfrak{p}$, and hence
$$ \zeta_{R_\eta}^p \equiv \tau_u(\zeta_{R_\eta})=\zeta_{R_\eta}^{A_\eta}\pmod{\mathfrak{p}},$$
for any $\mathfrak{p}$ lying above our choice of $p$ and below the associated $\mathfrak{P}$, 
which implies $p\equiv A_\eta\pmod{R_\eta}$.

Recall that $\mathbb{F}_p^\times=\mathbb{Z}/(p)$ is generated by $g+(p)$ and that $\zeta_{R_\eta}+\mathfrak{p}$ generates the order-$R_\eta$ cyclic subgroup of $(\mathcal{O}_{K_\eta}/\mathfrak{p})^\times$, so $\zeta_{\ell}+\mathfrak{p}$ generates the order-$\ell$ cyclic subgroup $\mu_\ell$ of $(\mathcal{O}_{K_\eta}/\mathfrak{p})^\times$, which is contained in $(\mathbb{Z}/(p))^\times$ since $\ell\mid(p-1)$. Therefore there exists $0\ne c\in\mathbb{Z}_\ell$ such that
$$
\zeta_\ell\equiv g^{c(p-1)/\ell}\pmod{\mathfrak{p}}
$$
and we find that $x\mapsto x^{(p-1)/\ell}$ is a homomorphism $\mathbb{F}_p^\times\to\mu_\ell$ such that
$$
\pi_\nu^{(p-1)/\ell}\equiv g^{cu_{\nu}(p-1)/\ell}\pmod{p}.
$$
Since $a_i\in C_{\alpha_i}^r$, then $a_i\equiv g^{\beta_i r+\alpha_i}\pmod{p}$ for some integer $\beta_i$. Then notice that $(p-1)/\ell$ is even, we have
$$
\begin{aligned}
g^{(\beta_i r+\alpha_i)(p-1)/\ell}&\equiv a_i^{(p-1)/\ell}\\&=\varepsilon_i^{(p-1)/\ell}\prod_{\nu=1}^N\pi_\nu^{e_{\nu i}(p-1)/\ell}\\
&\equiv \prod_{\nu=1}^N g^{cu_\nu e_{\nu i}(p-1)/\ell}\\
&=g^{(p-1)/\ell\sum cu_\nu e_{\nu i}}\pmod{p},
\end{aligned}
$$
which implies
$$
(\beta_ir+\alpha_i)(p-1)/\ell\equiv (p-1)/\ell\sum cu_\nu e_{\nu i}\pmod{p-1}.
$$
Divide both sides by $(p-1)/\ell$, we obtain
$$
\alpha_i\equiv\beta_ir+\alpha_i\equiv \sum cu_\nu e_{\nu i}=cu\cdot v(a_i)\pmod{\ell}
$$
as desired.

\end{proof}

\subsubsection{Nonresidue and separation consequences}

By choosing suitable $u$, we can make $a_i$s non-$r$th powers and lie in distinct cyclotomic classes for infinitely many prime moduli $p$.
These results will be used in the general sufficient theorems.

\begin{cor}%
\label{cor:arith-consequences}
With the notation of Theorem~\ref{thm:kummer-labeling}:
\begin{enumerate}[label=(\roman*)]
\item If $u\cdot v(a_i)\ne0$ for all $i$, then there are infinitely many primes $p\equiv A_\eta\pmod{R_\eta}$ for which every $a_i$ is a non-$r$th power in $\mathbb{F}_p^\times$.
\item If $u\cdot v(a_1),\ldots,u\cdot v(a_m)$ are pairwise distinct, then there are infinitely many primes $p\equiv A_\eta\pmod{R_\eta}$ for which $a_1,\ldots,a_m$ lie in pairwise distinct order-$r$ cyclotomic classes of $\mathbb{F}_p^\times$.
\end{enumerate}
\end{cor}

\begin{proof}
For (i), Theorem~\ref{thm:kummer-labeling} gives $\alpha_i\not\equiv0\pmod\ell$ for every $i$.  Hence $a_i$ is not an $\ell$th power modulo $p$.  Since $\ell\mid r$, %
$a_i$ is not an $r$th power.

For (ii), the same theorem gives labels whose reductions modulo $\ell$ are the prescribed values multiplied by a common nonzero scalar.  Pairwise distinct reductions modulo $\ell$ imply pairwise distinct labels in $\mathbb{Z}_r$.
\end{proof}

\begin{cor}
\label{cor:hyperplane}
Let $a_1,\ldots,a_m\in\mathbb{Q}^\times$, and suppose $a_i/a_j\notin(\mathbb{Q}^\times)^\ell$ for every $i\ne j$.  If
\[
        \ell>\binom{m}{2},
\]
then there exists $u\in\mathbb{F}_\ell^N$ separating $a_1,\ldots,a_m$, i.e., $u\cdot v(a_i)$ are pairwise distinct.  Consequently, whenever $\ell\mid r$, the numbers $a_1,\ldots,a_m$ can be forced into pairwise distinct order-$r$ cyclotomic classes for infinitely many primes in the relevant congruence class.
\end{cor}

\begin{proof}
For every pair $i<j$, the equation
\[
        u\cdot(v(a_i)-v(a_j))=0
\]
defines a proper hyperplane in $\mathbb{F}_\ell^N$.  There are $\binom{m}{2}$ such hyperplanes, and their union has size at most $\binom{m}{2}\ell^{N-1}<\ell^N$.  Hence, some $u$ avoids all of them.
\end{proof}

\subsubsection{A signed congruence}

When negative coefficients occur, we use the signed Chebotarev congruence (i.e. $\eta=1$) in Theorem~\ref{thm:kummer-labeling}.  The following elementary observation identifies the label of $-1$ under this congruence.

\begin{lemma}
\label{lem:sign}
If $r$ is even and $q\equiv1+r\pmod{2r}$, then $q=rt+1$ with $t$ odd and
\[
        -1\in C_{r/2}^r\subset\mathbb{F}_q^\times.
\]
\end{lemma}

\begin{proof}
For primitive $g\in\mathbb{F}_q^\times$,
\[
        -1=g^{(q-1)/2}=g^{rt/2}=g^{r(t-1)/2+r/2}\in g^{r/2}C_0^r=C_{r/2}^r.
\]
\end{proof}

\subsubsection{Two-primary auxiliary criterion}

If $r$ has no odd prime divisor, then the preceding argument does not apply.  In the cases needed below with $r=2^M$, $M\ge2$, it is enough to force the relevant rational numbers to be non-fourth-powers.

\begin{prop}
\label{prop:two-primary}
Let $r=2^M$ with $M\ge2$, and put $K=\mathbb{Q}(\zeta_r)$.  Suppose $a_1,\ldots,a_t\in\mathbb{Q}^\times$ and there is a homomorphism
\[
        \eta:\langle[a_1],\ldots,[a_t]\rangle\subset K^\times/(K^\times)^4\longrightarrow\mu_4
\]
with $\eta([a_i])\ne1$ for every $1\le i\le t$ where $[a_i]$ is the image of $a_i\in K^\times$ in $(K^\times)^4$ and $\mu_4$ is the group of 4th roots of unity.  Then there are infinitely many primes $p\equiv1\pmod r$ for which all $a_i$ are non-$r$th powers modulo $p$.
\end{prop}

\begin{proof}
Consider the Kummer extension
\[
        L=K(\sqrt[4]{a_1},\ldots,\sqrt[4]{a_t}).
\]
The homomorphism $\eta$ corresponds by Kummer theory to an element of $\Gal(L/K)$ whose action on $\sqrt[4]{a_i}$ is multiplication by $\eta([a_i])$.  Chebotarev's density theorem gives infinitely many primes splitting completely in $K$, hence satisfying $p\equiv1\pmod r$, whose Frobenius automorphism is that element.  For such $p$, the fourth-power residue symbol of $a_i$ is $\eta([a_i])\ne1$, so $a_i$ is not a fourth, and hence not an $r$th power modulo $p$.%
\end{proof}

\subsection{Families with $k_1+k_2\le3$}
\label{sec:small-families}

In the following applications, we will use Theorem~\ref{thm:main-framework}. %
Thus, we will focus on the verification of the packing condition (*).

\subsubsection{The family $(1,0)$}

\begin{thm}
\label{thm:10}
For every $b\ge2$, there exist infinitely many sufficiently large primes $q\equiv1\pmod{b2^{b-1}}$ such that, with
\[
        n=\frac{q-1}{2^{b-1}},
\]
there is a perfect $(1,0)$-limited-magnitude $b$-burst error correcting code.
\end{thm}

\begin{proof}
In this case, $A=\{0,1\}$, $A^*=\{1\}$, $s=1$, $d=2$, and $r=b2^{b-1}$.  Every block $\Lambda_j(\delta)=\{1\}$ is a singleton.  By Lemma~\ref{lem:gamma-count}, the total number of singleton blocks is $r$, so they can be assigned bijectively to the classes of $\mathbb{Z}_r$.  Theorem~\ref{thm:main-framework} applies.
\end{proof}

\subsubsection{The family $(1,1)$}

\begin{thm}
\label{thm:11}
For every $b\ge2$, put $r=2b3^{b-1}$.  There exist infinitely many sufficiently large primes
\[
        q\equiv1+r\pmod{2r}
\]
such that, with
\[
        n=\frac{q-1}{2\cdot3^{b-1}},
\]
there is a perfect $(1,1)$-limited-magnitude $b$-burst error correcting code.
\end{thm}

\begin{proof}
Here $A=\{-1,0,1\}$ and $A^*=\{-1,1\}$. It is easy to verify that $\Lambda_j(\delta)=\{-1,1\}$. By Lemma~\ref{lem:sign}, $-1\in C_c^r$ with $c=r/2$,  so each $B(\Lambda_j(\delta))$ contributes the pair block $\{0,c\}$.  The number of such pair blocks is
\[
        \sum_{j=0}^{b-1}|\Gamma_j(\{-1,0,1\})|=b3^{b-1}=r/2.
\]
The group $\mathbb{Z}_r$ decomposes into exactly $r/2$ disjoint pairs $\{u,u+c\}$.  Theorem~\ref{thm:main-framework} applies.
\end{proof}

\subsubsection{The family $(2,0)$}
If a root $\delta$ has a non-singleton block, then we say $\delta$ is a non-singleton root.

\begin{lemma}
\label{lem:20-blocks}
For $(k_2,k_1)=(2,0)$, the non-singleton roots are exactly the roots in $\Gamma_j(\{0,1\})$, and each such root has the block $\{0,\alpha\}$, where $2\in C_\alpha^r$.  Their total number is
\[
        \tau=b2^{b-1}.
\]
\end{lemma}

\begin{proof}
A root $\delta$ admits both multipliers $1$ and $2$ exactly when $\delta$ and $2\delta$ both have all coordinates in $\{0,1,2\}$.  This is equivalent coordinate-wise to every coefficient of $\delta$ lying in $\{0,1\}$.  The count follows from Lemma~\ref{lem:gamma-count}.
\end{proof}

\begin{thm}
\label{thm:20}
For every $b\ge2$, put $r=2b3^{b-1}$.  There exist infinitely many sufficiently large primes $q\equiv1\pmod r$ such that, with
\[
        n=\frac{q-1}{2\cdot3^{b-1}},
\]
there is a perfect $(2,0)$-limited-magnitude $b$-burst error correcting code.
\end{thm}

\begin{proof}
Apply Corollary~\ref{cor:arith-consequences} with $\ell=3$ to $a=2$ to get infinitely many primes $q\equiv1\pmod r$ for which $2\notin C_0^r$.  Write $2\in C_\alpha^r$ with $\alpha\ne0$.  The graph on $\mathbb{Z}_r$ with edges $\{u,u+\alpha\}$ is a disjoint union of cycles and single edges, and hence has a matching of size at least $r/3$.  Since
\[
        \tau=b2^{b-1}\le \frac r3=2b3^{b-2},
\]
all pair blocks can be packed into disjoint translates of $\{0,\alpha\}$.  The remaining blocks are singletons.
\end{proof}

\subsubsection{The family $(3,0)$}

\begin{lemma}
\label{lem:30-blocks}
For $(k_2,k_1)=(3,0)$, assume $2\in C_\alpha^r$ and $3\in C_\beta^r$, with $0,\alpha,\beta$ pairwise distinct.  The non-singleton roots are exactly the roots in $\Gamma_j(\{0,1\})$, each with block $\{0,\alpha,\beta\}$.  Their total number is
\[
        \tau=b2^{b-1}.
\]
\end{lemma}

\begin{proof}
A root admits at least two of the multipliers $1,2,3$ only when every coordinate lies in $\{0,1\}$.  Indeed, each pair among $u,2u,3u$ lying in $\{0,1,2,3\}$ forces $u\in\{0,1\}$.  In that case, all three multipliers are admissible.  The count is Lemma~\ref{lem:gamma-count} with alphabet $\{0,1\}$.
\end{proof}

\begin{thm}
\label{thm:30}
For every $b\ge2$, put $r=3b4^{b-1}$.  There exist infinitely many sufficiently large primes $q\equiv1\pmod r$ such that, with
\[
        n=\frac{q-1}{3\cdot4^{b-1}},
\]
there is a perfect $(3,0)$-limited-magnitude $b$-burst error correcting code.
\end{thm}

\begin{proof}
Use $\ell=3$ and exponent vectors for $2,3,3/2$ relative to the primes $2,3$:
\[
        v(2)=(1,0),\qquad v(3)=(0,1),\qquad v(3/2)=(-1,1).
\]
The vector $u=(1,2)\in\mathbb{F}_3^2$ gives nonzero dot products with all three.  Hence $2$, $3$, and $3/2$ can be made non-$r$th powers, so $1,2,3$ lie in pairwise distinct order-$r$ classes.  Write the labels as $0,\alpha,\beta$.

Let $S=\{0,\alpha,\beta\}$.  Build a graph on $\mathbb{Z}_r$ by joining $u$ and $v$ whenever $(u+S)\cap(v+S)\ne\varnothing$.  Every vertex has degree at most $6$, so the independence number is at least $r/7$.  For $b\ge3$, just let $r/7\ge b2^{b-1}$; for $b=2$, choose $r=24$ and the same integer bound gives at least $4=b2^{b-1}$ independent translates.  Thus, all triple blocks can be packed, and the remaining blocks are singletons.
\end{proof}

\subsubsection{The family $(2,1)$}

\begin{lemma}
\label{lem:pair-slot-matching}
Let $r$ be even, let $c=r/2$, and let $a\in\mathbb{Z}_c$ be nonzero.  Partition $\mathbb{Z}_r$ into pair slots
\[
        P_u=\{u,u+c\},\qquad u\in\mathbb{Z}_c.
\]
Let $G_a$ be the graph on $\mathbb{Z}_c$ with edges $\{u,u+a\}$.  Then $G_a$ has a matching of size at least $c/3$.  Consequently, if $\tau,p\ge0$ satisfy
\[
        \tau\le c/3,
        \qquad c-2\tau\ge p,
\]
then one can choose $\tau$ pairwise disjoint edge slots of the form $P_u\cup\{u+a\}$, and $p$ additional pairwise disjoint unused pair slots $P_v$.  The same conclusion holds with edge slots $P_u\cup P_{u+a}$.
\end{lemma}

\begin{proof}
The translation $u\mapsto u+a$ partitions $\mathbb{Z}_c$ into orbits.  Each component of $G_a$ is either a cycle or a single edge, and every component of size $L$ has a matching of size at least $L/3$.  Hence $G_a$ has a matching of size at least $c/3$.  A submatching of size $\tau$ touches exactly $2\tau$ pair-slot indices, leaving at least $c-2\tau$ untouched pair slots for the remaining $p$ pair blocks.  The same argument applies to the doubled edge slots $P_u\cup P_{u+a}$.
\end{proof}
\begin{lemma}
\label{lem:21-blocks}
For $(k_2,k_1)=(2,1)$, assume $-1\in C_c^r$ and $2\in C_\alpha^r$ with $\alpha\notin\{0,c\}$.  There are three types of blocks:
\begin{enumerate}[label=(\roman*)]
\item triple blocks $\{0,c,\alpha\}$, occurring exactly for roots in $\Gamma_j(\{0,1\})$;
\item pair blocks $\{0,c\}$, occurring exactly for roots in $\Gamma_j(\{-1,0,1\})\setminus\Gamma_j(\{0,1\})$;
\item singleton blocks, for all remaining roots.
\end{enumerate}
The total number of triple and pair blocks is
\[
        \tau=b2^{b-1},
        \qquad p_1=b3^{b-1}-b2^{b-1},
\]
respectively.
\end{lemma}

\begin{proof}
For a coordinate value $u$, inspect
\[
        S(u)=\{\lambda\in\{-1,1,2\}:\lambda u\in\{-1,0,1,2\}\}.
\]
One has $S(u)=\{-1,1,2\}$ exactly for $u\in\{0,1\}$, $S(u)=\{-1,1\}$ exactly for $u=-1$, and all other cases are singletons.  Intersecting these conditions coordinate-wise gives the three root types.  The counts follow from Lemma~\ref{lem:gamma-count}.
\end{proof}

\begin{thm}
\label{thm:21}
For every $b\ge2$, put $r=3b4^{b-1}$.  There exist infinitely many sufficiently large primes
\[
        q\equiv1+r\pmod{2r}
\]
such that, with
\[
        n=\frac{q-1}{3\cdot4^{b-1}},
\]
there is a perfect $(2,1)$-limited-magnitude $b$-burst error correcting code.
\end{thm}

\begin{proof}
Apply Corollary~\ref{cor:arith-consequences} in the signed congruence class with $\ell=3$ to the rational numbers $2$ and $-2$.  Then $2,-2\notin C_0^r$, and Lemma~\ref{lem:sign} gives $-1\in C_c^r$ with $c=r/2$.  If $2\in C_c^r$, then $-2\in C_0^r$, contradiction; hence $2\in C_\alpha^r$ with $\alpha\notin\{0,c\}$.

By Lemma~\ref{lem:21-blocks}, $\tau=b2^{b-1}$ and $p_1=b3^{b-1}-b2^{b-1}$.  Since
\[
        \tau\le c/3,
        \qquad c-2\tau\ge p_1
        \qquad (c=r/2),
\]
Lemma~\ref{lem:pair-slot-matching} packs all triple and pair blocks.  The remaining blocks are singletons, so Theorem~\ref{thm:main-framework} applies.
\end{proof}

\subsection{Families with $k_1 + k_2 = 4$}
\label{sec:four-families}

\subsubsection{The family $(4,0)$}

For this case, $A=\{0,1,2,3,4\}$ and $r=4b5^{b-1}$.

\begin{prop}
\label{prop:arith40}
For every $b\ge2$, there exist infinitely many primes $q\equiv1\pmod r$ such that if
\[
        2\in C_\alpha^r,
        \qquad 3\in C_\beta^r
\]
in $\mathbb{F}_q^\times$, then
\[
        \alpha\not\equiv0\pmod5,
        \qquad \beta\equiv3\alpha\pmod5.
\]
Consequently, the labels $0,\alpha,2\alpha,\beta$ of $1,2,4,3$ are pairwise distinct in $\mathbb{Z}_r$.
\end{prop}

\begin{proof}
Use Theorem~\ref{thm:kummer-labeling} with $\ell=5$ and exponent vectors $v(2)=(1,0)$, $v(3)=(0,1)$.  Choose $u=(1,3)$, then the labels of $2$ and $3$ modulo $5$ are in the ratio $1:3$.  Therefore $0,\alpha,2\alpha,\beta$ reduce modulo $5$ to $0,d,2d,3d$ with $d\ne0$.
\end{proof}

\begin{lemma}
\label{lem:40-blocks}
Let $2\in C_\alpha^r$ and $3\in C_\beta^r$.  There are four types of blocks:
\begin{enumerate}[label=(\roman*)]
\item four-blocks $\{0,\alpha,2\alpha,\beta\}$, occurring exactly for roots in $\Gamma_j(\{0,1\})$;
\item pair blocks $\{0,\alpha\}$, occurring exactly for roots in $\Gamma_j(\{0,1,2\})\setminus\Gamma_j(\{0,1\})$;
\item pair blocks $\{\alpha,2\alpha\}$, occurring exactly for roots in $\frac12\Gamma_j(\{0,1,2\})\setminus\Gamma_j(\{0,1\})$;
\item singleton blocks, for all remaining roots.
\end{enumerate}
Thus, the number of four-blocks is
\[
        \tau=b2^{b-1},
\]
and the total number of pair blocks is
\[
        p_0=2\bigl(b3^{b-1}-b2^{b-1}\bigr).
\]
\end{lemma}

\begin{proof}
For a coordinate value $u$, define
\[
        S(u)=\{\lambda\in\{1,2,3,4\}:\lambda u\in\{0,1,2,3,4\}\}.
\]
A direct inspection gives $S(u)=\{1,2,3,4\}$ exactly for $u\in\{0,1\}$, $S(u)=\{1,2\}$ exactly for $u=2$, $S(u)=\{2,4\}$ exactly for $u=1/2$, and all other cases are singletons.  Intersecting coordinate-wise gives the four families. The counts use Lemma~\ref{lem:gamma-count} for the alphabets $\{0,1\}$ and $\{0,1,2\}$.
\end{proof}

\begin{lemma}
\label{lem:40-packing}
Let $r=4b5^{b-1}$ and suppose $\alpha,\beta\in\mathbb{Z}_r$ satisfy
\[
        \alpha\not\equiv0\pmod5,
        \qquad \beta\equiv3\alpha\pmod5.
\]
Set $S=\{
0,\alpha,2\alpha,\beta\}$ and $P=\{0,\alpha\}$.  Then $\mathbb{Z}_r$ contains $\tau=b2^{b-1}$ pairwise disjoint translates of $S$ and, disjoint from them, $p_0=2(b3^{b-1}-b2^{b-1})$ pairwise disjoint translates of $P$.
\end{lemma}

\begin{proof}
Let $T\subset\mathbb{Z}_r$ be the set of residues congruent to $0$ modulo $5$, so $|T|=r/5=4b5^{b-2}$.  If base points of the translations are chosen from $T$, then reduction modulo $5$ separates the points of $S$ into residues $0,d,2d,3d$ and the points of $P$ into residues $0,d$, with $d=\alpha\bmod5\ne0$.  Hence, distinct base points in $T$ give disjoint translates, including mixed $S$--$P$ intersections.  It remains only to check
\[
        \tau+p_0\le r/5,
\]
which is equivalent to
\[
        2\cdot3^{b-1}-2^{b-1}\le4\cdot5^{b-2}.
\]
To prove it, notice that if 
\[
        F_b=4\cdot5^{b-2}-2\cdot3^{b-1}+2^{b-1},
\]
then one has $F_2=0$ and since
\[
        F_{b+1}=5F_b+4\cdot3^{b-1}-3\cdot2^{b-1}>0,
\]
then $F_b\ge0$ holds for all $b\ge2$.
\end{proof}

\begin{thm}
\label{thm:40}
For every $b\ge2$, put $r=4b5^{b-1}$.  There exist infinitely many primes $q\equiv1\pmod r$ such that, with
\[
        n=\frac{q-1}{4\cdot5^{b-1}},
\]
there is a perfect $(4,0)$-limited-magnitude $b$-burst error correcting code.
\end{thm}

\begin{proof}
Choose $q$ as in Proposition~\ref{prop:arith40}.  Lemma~\ref{lem:40-blocks} describes all non-singleton blocks, and Lemma~\ref{lem:40-packing} packs them.  The remaining classes are assigned to singleton blocks.  Theorem~\ref{thm:main-framework} applies.
\end{proof}

\subsection{The family $(3,1)$}

For this subsection, $A=\{-1,0,1,2,3\}$ and again $r=4b5^{b-1}$.

\begin{prop}
\label{prop:arith31}
For every $b\ge2$, there exist infinitely many primes
\[
        q\equiv1+r\pmod{2r}
\]
such that if $2\in C_\alpha^r$ and $3\in C_\beta^r$, then
\[
        \alpha\not\equiv0\pmod5,
        \qquad \beta\equiv2\alpha\pmod5.
\]
Moreover, with $c=r/2$, one has $-1\in C_c^r$, and $0,c,\alpha,\beta$ are pairwise distinct.
\end{prop}

\begin{proof}
Use the signed version of Theorem~\ref{thm:kummer-labeling} with $\ell=5$ and $u=(1,2)$ on the exponent vectors of $2$ and $3$.  This gives fifth-power labels in the ratio $1:2$, up to a common nonzero scalar, while the signed congruence gives $-1\in C_{r/2}^r$.  Since $r/2\equiv0\pmod5$, the residues of $0,c,\alpha,\beta$ modulo $5$ are $0,0,d,2d$ with $d\ne0$, and $0\ne c$ in $\mathbb{Z}_r$.
\end{proof}

\begin{lemma}
\label{lem:31-blocks}
Let $c=r/2$, $-1\in C_c^r$, $2\in C_\alpha^r$, and $3\in C_\beta^r$.  There are three types of blocks:
\begin{enumerate}[label=(\roman*)]
\item four-blocks $\{0,c,\alpha,\beta\}$, occurring exactly for roots in $\Gamma_j(\{0,1\})$;
\item pair blocks $\{0,c\}$, occurring exactly for roots in $\Gamma_j(\{-1,0,1\})\setminus\Gamma_j(\{0,1\})$;
\item singleton blocks, for all remaining roots.
\end{enumerate}
The number of four-blocks is $\tau=b2^{b-1}$, and the number of pair blocks is
\[
        p_1=b3^{b-1}-b2^{b-1}.
\]
\end{lemma}

\begin{proof}
For a coordinate value $u$, inspect
\[
        S(u)=\{\lambda\in\{-1,1,2,3\}:\lambda u\in\{-1,0,1,2,3\}\}.
\]
A direct check gives $S(u)=\{-1,1,2,3\}$ exactly for $u\in\{0,1\}$, $S(u)=\{-1,1\}$ exactly for $u=-1$, and singleton cases otherwise.  The counts follow from Lemma~\ref{lem:gamma-count}.
\end{proof}

\begin{lemma}
\label{lem:31-packing}
Let $r=4b5^{b-1}$, $c=r/2$, and suppose
\[
        \alpha\not\equiv0\pmod5,
        \qquad \beta\equiv2\alpha\pmod5.
\]
Then $\mathbb{Z}_r$ contains $\tau=b2^{b-1}$ pairwise disjoint translates of $S=\{0,c,\alpha,\beta\}$ and, disjoint from them, $p_1=b3^{b-1}-b2^{b-1}$ pairwise disjoint translates of $P=\{0,c\}$.
\end{lemma}

\begin{proof}
Pair slots $P_u=\{u,u+c\}$ lie inside single residue classes modulo $5$ because $c\equiv0\pmod5$.  Choose all four-block base pair slots in one residue class, say residue $0$.  There are $(r/2)/5=2b5^{b-2}$ such pair slots, and
\[
        b2^{b-1}\le2b5^{b-2}.
\]
A four-block based at residue $0$ uses residue classes $0,d,2d$, where $d=\alpha\bmod5\ne0$, and distinct base pair slots do not collide.  The pair blocks are placed in the unused residue classes $3d$ and $4d$, which together contain $4b5^{b-2}$ pair slots.  The needed inequality
\[
        b3^{b-1}-b2^{b-1}\le4b5^{b-2}
\]
follows from $3^{b-1}\le 3\cdot5^{b-2}\le4\cdot5^{b-2}$ for $b\ge2$.
\end{proof}

\begin{thm}
\label{thm:31}
For every $b\ge2$, put $r=4b5^{b-1}$.  There exist infinitely many sufficiently large primes
\[
        q\equiv1+r\pmod{2r}
\]
such that, with
\[
        n=\frac{q-1}{4\cdot5^{b-1}},
\]
there is a perfect $(3,1)$-limited-magnitude $b$-burst error correcting code.
\end{thm}

\begin{proof}
Choose $q$ as in Proposition~\ref{prop:arith31}.  Lemma~\ref{lem:31-blocks} gives the block types and counts, and Lemma~\ref{lem:31-packing} supplies the disjoint translations.  Theorem~\ref{thm:main-framework} applies.
\end{proof}

\subsection{The family $(2,2)$}

For this subsection, $A=\{-2,-1,0,1,2\}$ and $r=4b5^{b-1}$, $c=r/2$.

\begin{lemma}
\label{lem:22-blocks}
Assume $-1\in C_c^r$, $2\in C_\alpha^r$, and $\alpha\notin\{0,c\}$.  There are three types of blocks:
\begin{enumerate}[label=(\roman*)]
\item four-blocks $\{0,c,\alpha,\alpha+c\}$, occurring exactly for roots in $\Gamma_j(\{-1,0,1\})$;
\item pair blocks $\{0,c\}$, occurring exactly for roots in $\Gamma_j(\{-2,-1,0,1,2\})\setminus\Gamma_j(\{-1,0,1\})$;
\item pair blocks $\{\alpha,\alpha+c\}$, occurring exactly for roots in $\frac12\Gamma_j(\{-2,-1,0,1,2\})\setminus\Gamma_j(\{-1,0,1\})$.
\end{enumerate}
There are no singleton roots.  The number of four-blocks is
\[
        \tau=b3^{b-1},
\]
and the total number of pair blocks is
\[
        p_2=2\bigl(b5^{b-1}-b3^{b-1}\bigr).
\]
\end{lemma}

\begin{proof}
For a coordinate value $u$, inspect
\[
        S(u)=\{\lambda\in\{-2,-1,1,2\}:\lambda u\in\{-2,-1,0,1,2\}\}.
\]
A direct check gives $S(u)=\{-2,-1,1,2\}$ exactly for $u\in\{-1,0,1\}$, $S(u)=\{-1,1\}$ exactly for $u\in\{-2,2\}$, and $S(u)=\{-2,2\}$ exactly for $u\in\{-1/2,1/2\}$.  Intersecting coordinate-wise gives the three families.  The counts follow from Lemma~\ref{lem:gamma-count}, and the two pair families have the same size by scaling by $2$.
\end{proof}

\begin{thm}
\label{thm:22}
For every $b\ge2$, put $r=4b5^{b-1}$.  There exist infinitely many sufficiently large primes
\[
        q\equiv1+r\pmod{2r}
\]
such that, with
\[
        n=\frac{q-1}{4\cdot5^{b-1}},
\]
there is a perfect $(2,2)$-limited-magnitude $b$-burst error correcting code.
\end{thm}

\begin{proof}
Use the signed form of Corollary~\ref{cor:arith-consequences} with $\ell=5$ for the rational numbers $2$ and $-2$.  Then $2,-2\notin C_0^r$, while Lemma~\ref{lem:sign} gives $-1\in C_c^r$, where $c=r/2$.  If $2\in C_c^r$, then $-2\in C_0^r$, contradiction; hence $2\in C_\alpha^r$ with $\alpha\notin\{0,c\}$.

By Lemma~\ref{lem:22-blocks}, $\tau=b3^{b-1}$ and $p_2=2(b5^{b-1}-b3^{b-1})$.  Since
\[
        \tau\le c/3,
        \qquad c-2\tau=p_2,
\]
Lemma~\ref{lem:pair-slot-matching} supplies enough disjoint $4$-slots $P_u\cup P_{u+\alpha}$ and remaining pair slots.  Theorem~\ref{thm:main-framework} applies.
\end{proof}

\subsection{General sufficient theorems for limited magnitude burst correcting codes}
\label{sec:general}

\subsubsection{The family $(k,0)$}

Let $k\ge2$, $A=[0,k]$, $s=k$, $d=k+1$, and $r=bsd^{b-1}$.  Put
\[
        M_k=\left\lfloor\frac{k}{2}\right\rfloor+1.
\]

\begin{lemma}
\label{lem:k0-count}
The number of non-singleton cyclotomic blocks for $(k_2,k_1)=(k,0)$ is at most
\[
        \binom{k}{2} bM_k^{b-1}.
\]
\end{lemma}

\begin{proof}
If a root admits two distinct multipliers $\lambda<\mu$, then each coordinate $u$ of this root must satisfy both $\lambda u\in\{0,\ldots,k\}$ and $\mu u\in\{0,\ldots,k\}$.  The number of such coordinate values is at most $\lfloor k/(\mu/g)\rfloor+1\le\lfloor k/2\rfloor+1$, where $g=\gcd(\lambda,\mu)$.  For each pair of multipliers, Lemma~\ref{lem:gamma-count} bounds the number of roots supported on such an alphabet by $bM_k^{b-1}$.
\end{proof}

\begin{lemma}
\label{lem:greedy}
Let $\mathcal B$ be a family of at most $N$ subsets of $\mathbb{Z}_r$, each of size at most $s$.  If one block may be fixed and all others may be translated, and if
\[
        s^2N<r,
\]
then the blocks can be translated so that they are pairwise disjoint.
\end{lemma}

\begin{proof}
Place the blocks greedily.  Before placing a block $B$, at most $sN$ points of $\mathbb{Z}_r$ have been covered.  A translation $u+B$ is bad only if $u+b$ equals an already covered point for some $b\in B$, giving at most $s^2N$ bad choices of $u$.  Since $s^2N<r$, a good translate remains.
\end{proof}

\begin{thm}
\label{thm:k0}
Let $k\ge2$, $b\ge2$, and $r=bk(k+1)^{b-1}$.  Assume that some odd prime $\ell\mid r$ and some $u\in\mathbb{F}_\ell^N$ makes
\[
        u\cdot v(1),\ldots,u\cdot v(k)
\]
pairwise distinct.  If
\[
        k\binom{k}{2}
        \left(\frac{\lfloor k/2\rfloor+1}{k+1}\right)^{b-1}<1,
\]
then there exist infinitely many sufficiently large primes $q\equiv1\pmod r$ such that, with
\[
        n=\frac{q-1}{k(k+1)^{b-1}},
\]
there is a perfect linear code in $\mathbb{Z}^n$ correcting a single cyclic $b$-burst of $(k,0)$-limited-magnitude errors.
\end{thm}

\begin{proof}
The arithmetic hypothesis and Corollary~\ref{cor:arith-consequences} separate the labels of $1,\ldots,k$.  By Lemma~\ref{lem:k0-count}, the number of non-singleton blocks is at most $N=\binom{k}{2}bM_k^{b-1}$, each of size at most $k$.  The displayed inequality is exactly $k^2N<r$, so Lemma~\ref{lem:greedy} packs all non-singleton blocks.  The remaining blocks are singletons and fill the remaining classes by Lemma~\ref{lem:gamma-count}.  Theorem~\ref{thm:main-framework} applies.
\end{proof}

\begin{cor}
\label{cor:k0-checkable}
The separating hypothesis of Theorem~\ref{thm:k0} is automatically satisfied if $r$ has an odd prime divisor
\[
        \ell>\binom{k}{2}
\]
and no ratio among $1,\ldots,k$ is an $\ell$th power in $\mathbb{Q}^\times$.
\end{cor}

\begin{proof}
Apply Corollary~\ref{cor:hyperplane} to the rational numbers $1,\ldots,k$.
\end{proof}

\subsubsection{Asymmetric and symmetric families}

Let $A=[-k_1,k_2]$, $s=k_2+k_1$, $d=s+1$, and $r=bs(s+1)^{b-1}$.

\begin{thm}
\label{thm:asym}
Assume $k_2>k_1\ge0$.  Let
\[
        M_A=\max_{\lambda\ne\mu,
        \;\lambda,\mu\in A^*}
        |\lambda^{-1}A\cap\mu^{-1}A|.
\]
Suppose some odd prime $\ell\mid r$ and some $u\in\mathbb{F}_\ell^N$ make
\[
        u\cdot v(1),\ldots,u\cdot v(k_2)
\]
pairwise distinct.  If
\[
        s\binom{s}{2}\left(\frac{M_A}{s+1}\right)^{b-1}<1,
\]
then there exist infinitely many sufficiently large primes $q$ satisfying $q\equiv1\pmod r$ when $k_1=0$, and $q\equiv1+r\pmod{2r}$ when $k_1>0$, such that, with
\[
        n=\frac{q-1}{s(s+1)^{b-1}}
\]
there is a perfect $(k_2,k_1)$-limited-magnitude $b$-burst error correcting code.
\end{thm}

\begin{proof}
The arithmetic condition separates the positive coefficient labels.  If $k_1>0$, the signed congruence places $-m$ in the label of $m$ shifted by $r/2$.  A root admitting two distinct multipliers $\lambda,\mu$ has all coordinates in $\lambda^{-1}A\cap\mu^{-1}A$, so the number of non-singleton blocks is at most $\binom{s}{2}bM_A^{b-1}$.  Each block has size at most $s$, and the displayed inequality is the greedy condition $s^2N<r$.  Apply Lemma~\ref{lem:greedy} and Theorem~\ref{thm:main-framework}.
\end{proof}

\begin{thm}
\label{thm:sym}
Let $(k_2,k_1)=(K,K)$ and
\[
        r=2bK(2K+1)^{b-1}.
\]
Suppose some odd prime $\ell\mid r$ and some $u\in\mathbb{F}_\ell^N$ make
\[
        u\cdot v(1),\ldots,u\cdot v(K)
\]
pairwise distinct.  If
\[
        K\binom{K}{2}
        \left(\frac{2\lfloor K/2\rfloor+1}{2K+1}\right)^{b-1}<1,
\]
then there exist infinitely many sufficiently large primes
\[
        q\equiv1+r\pmod{2r}
\]
such that, with
\[
        n=\frac{q-1}{2K(2K+1)^{b-1}},
\]
there is a perfect $(K,K)$-limited-magnitude $b$-burst error correcting code.
\end{thm}

\begin{proof}
The signed congruence makes $\{u,u+r/2\}$ the natural sign-pair slot.  Passing to the quotient by these slots leaves only the positive multipliers $1,\ldots,K$.  If two positive multipliers are simultaneously admissible at a coordinate, then that coordinate has at most $2\lfloor K/2\rfloor+1$ possible values.  Hence the quotient non-singleton blocks are bounded by $\binom{K}{2}b(2\lfloor K/2\rfloor+1)^{b-1}$ and have size at most $K$.  The displayed inequality is exactly the greedy packing condition in the quotient.  Lifting the packing to $\mathbb{Z}_r$ and applying Theorem~\ref{thm:main-framework} gives the result.
\end{proof}
\section{Conclusion}
In this work, we have explored the theory of limited magnitude error-correcting codes through a detailed study of splitter sets, group splittings, and lattice tilings. Our results include a complete classification of nonsingular quasi-perfect $B[0,3](n)$ for all $n$ and $B[-4,4](2p)$ sets for prime $p$. For perfect splitter sets, the existence condition of perfect $B[0,6](q)$ sets for prime $q$ is determined. By using tools from graph theory, we also improve lower bounds on $M(0,3;q)$ for prime $q$. Moreover, we give a general framework to obtain perfect limited magnitude burst correcting codes and construct new infinite families of them. %
Future work may focus on tightening bounds for $M(0,3;q)$ or noncyclic burst codes.
\section*{Acknowledgements}
We thank Xiaoxiao Li and Yibing Chen for checking some proofs of this manuscript.
\bibliographystyle{IEEEtrans}
\bibliography{splitter.bib}
\end{document}